\documentclass{article}

\usepackage{xargs}                     \usepackage[colorinlistoftodos,prependcaption,textsize=tiny]{todonotes}

\usepackage[normalem]{ulem}

\usepackage[utf8]{inputenc} 
\usepackage[T1]{fontenc}    
\usepackage[breaklinks=true,colorlinks = true,citecolor=blue,linkcolor=blue]{hyperref}
\usepackage{url}            
\usepackage{booktabs,comment}       
\usepackage{amsfonts,graphics, graphicx}       
\usepackage{nicefrac,amsmath,amsthm, amssymb,mathrsfs,wrapfig}       
\usepackage{microtype,xcolor}      
\usepackage[noend,ruled,linesnumbered]{algorithm2e}
\SetKwProg{Fn}{Function}{}{}
\DontPrintSemicolon
\usepackage{soul}
\newcommand{\argmin}{\operatornamewithlimits{argmin}}

\usepackage{cleveref}

\usepackage{amsfonts,amsmath,amsthm,dsfont,comment,mathpazo,xcolor}
\usepackage{cleveref}

\usepackage[margin=1in]{geometry}

\newcommand{\exclude}[1]{}

\usepackage{xargs}                     \usepackage[colorinlistoftodos,prependcaption,textsize=tiny]{todonotes}
\newcommand{\monika}[1]{\todo[linecolor=green,backgroundcolor=green!25,bordercolor=green]{{\bf Monika:} #1}}
\newcommand{\jalaj}[1]{\todo[linecolor=blue,backgroundcolor=blue!25,bordercolor=blue]{{\bf Jalaj:} :#1}}

\newcommand{\uniformBound}{C_{\frac{\epsilon}2,\frac{\delta}2} \sqrt{\frac{ \ln(6(\epsilon\sqrt{n}+1))}{2n}} \paren{1 +  \frac{\ln ((4\epsilon\sqrt{n} -3))/5}{\pi}}}

\newcommand{\removed}[1]{}

\newtheorem{myremark}{Remark}
\newenvironment{remark}{\begin{myremark}}{\end{myremark}}

\newtheorem{myexample}{Example}

\newcommand{\brak}[1]{{\left\langle {#1} \right\rangle}}

\newcommand{\set}[1]{\left\{{#1} \right\}}
\newcommand{\paren}[1]{\left( {#1} \right)}
\newcommand{\sparen}[1]{\left[ {#1} \right]}

\renewcommand{\leq}{\leqslant}
\renewcommand{\geq}{\geqslant}
\newcommand{\ceil}[1]{\lceil {#1} \rceil}
\newcommand{\floor}[1]{\lfloor {#1} \rfloor}

\newcommand{{\R}}{\mathbb{R}}
\newcommand{{\I}}{\mathds{1}}

\newtheorem{theorem}{Theorem}
\newtheorem{lemma}{Lemma}

\newtheorem{corollary}{Corollary}

\newtheorem{definition}{Definition} 
\newcommand{\episode}{M_\mathsf{ep}}

\newcommand{\hist}{M_{\mathsf{hist}}}

\def\theth/{\textsuperscript{th}}

\newcommand{\cb}{\mathsf{cb}}

\newcommand{\dist}{\mathsf{dist}}

\newcommand{\stu}{\mathsf{STU}}

\newcommand{{\bQ}}{\mathbf{\Psi}}

\newcommand{\cC}{\mathcal{C}}
\newcommand{\cD}{\mathcal{D}}

\newcommand{\cG}{\mathcal{G}}

\newcommand{\cL}{\mathcal{L}}

\newcommand{\cN}{\mathcal{N}}

\newcommand{\cP}{\mathcal{P}}
\newcommand{\cQ}{\mathcal{Q}}
\newcommand{\cR}{\mathcal{R}}

\newcommand{\cX}{\mathcal{X}}

\newcommand{\Med}{\mathsf{med}}

\newcommand{\citet}{\cite}
\newcommand{\pr}{\mathsf{Pr}}
\newcommand{\E}{\mathsf{E}}

\newcommand{\average}{M_\mathsf{average}}
\newcommand{\counting}{M_\mathsf{count}}
\newcommand{\slidecount}{M_\mathsf{slide\mbox{-}count}}
\newcommand{\cut}{M_\mathsf{cut}}

\newcommand{\eps}{\varepsilon}

\newcommand{\norm}[1]{\left\| #1 \right\|}

\usepackage{silence}

\graphicspath{{./graphics/}}
\usepackage{wrapfig}
\usepackage{calc}

\newcommand{\papertitle}{Constant matters:  Fine-grained Error Bound on Differentially Private Continual Observation Using Completely Bounded Norms}

\title{\papertitle}

\author{
Hendrik Fichtenberger \thanks{Google Research, Zurich. Email: \small{\sf fichtenberger@google.com}}
\and 
Monika Henzinger \thanks{University of Vienna. A part of the work was done as the Stanford University Distinguished Visiting Austrian Chair.
Email: \small{\sf monika.henzinger@univie.ac.at}
}
\and 
Jalaj Upadhyay \thanks{Rutgers University.  Email: \small{\sf jalaj.upadhyay@rutgers.edu}}
}

\date{}
\begin{document}

\maketitle
\begin{abstract}
   We study fine-grained error bounds for differentially private algorithms for counting under continual observation. Our main insight is that the matrix mechanism, when using lower-triangular matrices, can be used in the continual observation model. More specifically, we give an explicit factorization for the counting matrix $M_\mathsf{count}$ and upper bound the error explicitly. We also give a  fine-grained analysis, specifying the exact constant in the upper bound. Our analysis is based on upper and lower bounds of the {\em completely bounded norm} (cb-norm) of $M_\mathsf{count}$. Furthermore, we are the first to give concrete error bounds for various problems under continual observation such as binary counting, maintaining a histogram, releasing an approximately cut-preserving synthetic graph, many graph-based statistics, and substring and episode counting. Finally we  note that our result can be used to get a fine-grained error bound for non-interactive local learning {and the first lower bounds on the additive error for $(\epsilon,\delta)$-differentially-private counting under continual observation.}
    Subsequent to this work, Henzinger et al. (SODA, 2023) showed that our factorization also achieves fine-grained mean-squared error.
\end{abstract}

\clearpage
\tableofcontents

\clearpage
\pagenumbering{arabic}

\section{Introduction}
In recent times, many large scale applications of data analysis involved repeated computations because of the incidence of infectious diseases~\cite{applecontact2021, cdccontact}, typically with the goal of preparing some appropriate response. However, privacy of the user data (such as a positive or negative test result) is equally important. In such an application, the system is required to continually produce outputs while preserving a robust privacy guarantee such as {\em differential privacy}. 
This setting was already used as motivating example by \citet{dwork2010differentially} in the first work on differential privacy under continual release, where they write:
\begin{quote}
``Consider a website for H1N1 self-assessment. Individuals can interact with the site to learn whether symptoms they are experiencing may be indicative of the H1N1 flu. The user fills in some demographic data (age, zipcode, sex), and responds to queries about his symptoms (fever over $100.4^\circ$ F?, sore throat?, duration of symptoms?). We would like to continually analyze aggregate information of consenting users in order to monitor regional health conditions, with the goal, for example, of organizing improved flu response. Can we do this in a differentially private fashion with reasonable accuracy (despite the fact that the system is continually producing outputs)?"
\end{quote}

In the \emph{continual release} (or \emph{observation}) \emph{model} the input data arrives as a stream of items $x_1, x_2, \dots, x_T$, with data $x_i$ arriving in {\em round $i$} and the mechanism has to be able to output an answer after the arrival of each item. The study of the continual release model was initiated concurrently by \citet{dwork2010differentially} and \citet{chan2011private} through the study of the \emph{(continual) binary counting} problem: Given a stream of bits, i.e., zeros and ones, output after each bit the sum of the bits so far, i.e., the number of ones in the input so far.
\citet{dwork2010differentially}  showed that  there exists a differentially private mechanism, called the \emph{binary (tree) mechanism}, for this problem 
with an additive error of $O(\log^{5/2} T)$, where the additive error is the maximum additive error over all rounds $i$. This was further improved to $O(\log^{5/2}(t))$ at time $t \leq T$ by \citet{chan2011private} (see their Corollary 5.3). These algorithms use Laplace noise for differential privacy. \citet{jain2021price} showed that using Gaussian noise instead an additive error of  $O(\log(t)\sqrt{\log(T)})$ can be achieved. 
However, the constant has never been explicitly stated. Given the wide applications of binary counting in many downstream tasks, such as counting in the sliding window model~\cite{bolot2013private}, frequency estimation~\cite{cardoso2021differentially},  graph problems~\cite{fichtenberger2021differentially}, frequency estimation in the sliding window model~\cite{chan2012differentially,epasto2023differentially, huang2021frequency,upadhyay2019sublinear}, counting under adaptive adversary~\cite{jain2021price}, optimization~\cite{choquette2022multi, mcmahan2022private, henzinger2022almost, kairouz2021practical, smith2017interaction}, graph spectrum~\cite{upadhyay2021differentially}, and matrix analysis~\cite{upadhyay2021framework}, constants can define whether the output is useful or not in practice. In fact, from the practitioner's point of view, it is important to know the constant hidden in 
the $O(\cdot)$ notation. 
With this in mind, we ask  the following central question:
\begin{quote}
    {\em Can we get fine-grained bounds on the constants 
    in the additive error of differentially private algorithms for binary counting under continual release?}
\end{quote}

The problem of reducing the additive error for binary  counting under continual release has been pursued before (see~\citet{wang2021continuous} and references therein).  Most of them use some ``smoothening" technique~\cite{wang2021continuous}, assume some structure in the data 
~\cite{rastogi2010differentially}, or measure error in mean squared loss~\cite{wang2021continuous}\footnote{
While mean squared error is useful in some applications like learning~\cite{mcmahan2022private, henzinger2022almost, kairouz2021practical}, in many application we prefer a worst case additive error, the metric of choice in this paper.}.
There is a practical reason to smoothen the output of the binary mechanism as its additive error is highly non-smooth (see \Cref{fig:running_count_00}) due to the way the binary mechanism works: its additive error at any time $t$ depends on how many dyadic intervals are summed up in the output for $t$. 
This forces the error to have non-uniformity of order $\log_2(T)$, which makes it hard to \emph{interpret}\footnote{
For example, consider a use case of ECG monitoring on a smartwatch of a heart patient. Depending on whether $t=2^{i}-1$ and $t=2^i$ for some $i \in \mathbb N$, the error of the output of binary mechanism might send an SOS signal or not.}.
For example in {\em exposure-notification systems} that has to operate when the stream length is in order of $10^8$,
it is desirable that the algorithm is scalable and its output 
fulfills properties such as monotonicity and  smoothness to make the output interpretable.
Thus, \emph{the focus of this paper is to design a scalable mechanism for binary counting in the continual release model with a smooth additive error and to show a (small) fine-grained error bound.} 

\begin{figure}[h]
\centering
\caption{\small{Additive $\ell_{\infty}$ error with $T=2^{16}, \epsilon = 0.8, \delta = 10^{-10}$.}}
\includegraphics[scale=0.7]{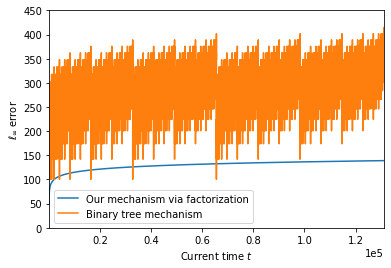}
    \label{fig:running_count_00}
\vspace{-5mm}
\end{figure}

\noindent {\bf Our contributions.}
We prove concrete bounds on the additive error for counting under continual release that is tight up to a small additive gap. Since our bounds are closed-form expressions, it is straightforward to evaluate them for any data analysis task. Furthermore, our algorithms only perform few matrix-vector multiplications and additions, which makes them easy to implement and tailor to operations natively implemented in modern hardware. Finally, our algorithms are also efficient and the
additive error of the output is smooth. 
As counting is a versatile building block, we get concrete bounds on a wide class of problems under continual release such as maintaining histograms, generating synthetic graphs that  preserve cut sizes, computing various graph functions,  and substring and counting episode in string sequences (see \Cref{sec:results} for more details). Furthermore, this also leads to an improvement in the additive error for non-interactive local learning~\cite{smith2017interaction}, private learning on non-Euclidean geometry~\cite{asi2021private}, and private online prediction~\cite{asi2022private}: These algorithms use the binary mechanism as a subroutine and using our mechanism instead would give a constant-factor improvements in the additive error. This in turn allows to \emph{decrease} the privacy parameter in private learning, which is highly desirable and listed as  motivation in several recent works~\cite{asi2022optimal, mcmahan2022private, henzinger2022almost}\footnote{Note that \citet{asi2022optimal} and \citet{henzinger2022almost} appeared subsequent to this work on arXiv.}.

Our bounds bridge the gap between theoretical work on differential privacy, which mostly concentrates on asymptotic analysis to reveal the \emph{capabilities of} differential private algorithms, and practical applications, which need to obtain useful \emph{information from} differential private algorithms for their specific use cases.

\noindent {\bf Organization.} The rest of this section gives the formal problem statement, an overview of results, technical contribution, and comparison with related works. 
\Cref{sec:notation} gives the necessary definitions, \Cref{sec:proof} presents the formal proof of our main result, \Cref{sec:applications} contains all its applications we explored. We give lower bounds in~\Cref{sec:lower} and present experiments in \Cref{sec:experiments}.

\subsection{The Formal Problem}
\emph{Linear queries} are classically defined as follows: 
There is a universe ${\cal X} = \{0,1\}^d$ of values and a set ${\mathcal Q} = \{q_1, \dots, q_k\}$ of functions
$q_i : {\cal X} \rightarrow \mathbb{R}$ with $1 \leq i \leq k$. 
Given a vector $x= (x[1], \dots, x[n])$ 
of $n$ values of $\cal X$ (with repetitions allowed) a
\emph{linear query $q(x)$ for the function $q$} computes 
$\sum_{j=1}^n q(x[j])$.\footnote{Usually a linear query is defined to return the value
$\frac{1}{n} \sum_{i=j}^n q(x_j)$, but as we assume that $n$ is publicly known it is simpler to use our formula.}
A \emph{workload} for a vector $x$ and a set $\{q_1, \dots, q_k\}$  of functions computes the 
linear query $q_i(x)$ for each function $q_i$ with $1 \leq i \leq k$. This computation
can be formalized using linear algebra notation as follows:
Assume there is a fixed ordering $y_1, \dots y_{2^d}$ of all elements of $\cal X$.
The \emph{workload matrix} $M$ is defined by $M[i,j] = q_i(y_j)$, i.e. there is a row for each function $q_i$ and a column for each value $y_j$. Let $h \in \mathbb{N}_0^{2^d}$ be
the histogram vector of $x$, i.e.~$y_j$ appears $h(y_j)$ times in $x$. Then answering the linear queries is equivalent to computing $M h$.

In the continual release setting, the vector $x$ is given incrementally to the mechanism in {\em rounds} or {\em time steps}. In time step $t$, $x[t]$ is revealed to the mechanism and it has to output $M_t x$ under differential privacy, where $M$ is the workload matrix and $M_t$ denotes the  $t \times t$ principal submatrix of $M$.

Binary counting corresponds to a very simple linear query in this setting: The universe ${\cal X}$ equals $\{0,1\}$, there is only one query $q:{\cal X} \rightarrow \mathbb{R}$ with $q(1) = 1$ and $q(0) = 0$. However, alternatively, binary counting could also be expressed as follows and this is the notation that we will use: There is only one query $q'$ with $q'(y) = 1$ for all $y \in {\cal X}$ giving rise to a simple workload matrix $M = (1, \dots, 1)$ for the static setting and the mechanism outputs $Mx$.
%
Thus, in the continual release setting, we study the following the following workload matrix, $\counting \in \{0,1\}^{T \times T}$:
\begin{align}
    \counting[i,j] = \begin{cases}
    1 & i \geq j \\
    0 & i <j
    \end{cases}
    \label{eq:matrices}
\end{align}
where for any matrix $A$, $A[i,j]$ denote its $(i,j)$\theth/ coordinate.

There has been a large body of work on designing differentially private algorithms for general workload matrices in the static setting, i.e.,~not under continual release. 
One of the scalable techniques that provably reduces the error on linear queries is a query matrix optimization technique known as \emph{workload optimizer} (see~\citet{mckenna2021hdmm} and references therein). There have been various algorithms developed for this, one of them being the {\em matrix mechanism}~\cite{li2015matrix}, which first determines two matrices $R$ and $L$ such that $M=LR$ and then outputs $L (R x+z)$, where $z \sim N(0,\sigma^2 \mathbb I)$ is a vector of Gaussian values for a suitable choice of $\sigma^2$ and $\mathbb I$ is the identity matrix.
For a privacy budget $(\epsilon,\delta)$, it can be shown that the additive error, denoted as $\ell_\infty$ error of the answer vector (see \Cref{def:accuracy}), of the matrix mechanism for $|\mathcal Q|$ linear queries with $\ell_2$-sensitivity $\Delta_{\mathcal Q}$ (\cref{eq:ell_2sensitivity}) represented by a workload matrix $M$ using the Gaussian mechanism is as follows: with probability $2/3$ over the random coins of the algorithm, the additive error is at most
\begin{align}
\label{eq:error_bound}
\underbrace{\frac{2}{\epsilon} \sqrt{{\frac{4}{9} +  \ln \paren{\frac{1}{\delta}\sqrt{\frac{2}{\pi}}}}}}
_{C_{\epsilon,\delta}} 
\Delta_{\mathcal Q} \norm{L}_{2 \to \infty}\norm{R}_{1 \to 2}\sqrt{ \ln (6|\mathcal Q|)}, 
\vspace{-5mm}
\end{align}
where $\norm{A}_{2 \to \infty}$ (resp., $\norm{A}_{1 \to 2}$)   is the maximum 
$\ell_2$ norm of  
rows (resp.~columns) of $A$. 
The function $C_{\epsilon,\delta}$ 
is a function 
arising as in the proof of the privacy guarantee
of the Gaussian mechanism when $\epsilon <1$ (see Theorem A.1 in~\citet{dwork2014algorithmic}. 
If  $\epsilon \geq 1$, one can analytically compute $C_{\epsilon,\delta}$  using Algorithm 1 in~\citet{balle2018improving}. 
For the rest of this paper, we use $C_{\epsilon,\delta}$ to denote this function.
Therefore, we need to find a factorization $M=LR$ that minimizes $\norm{L}_{2 \to \infty}\norm{R}_{1 \to 2}$.
Note that the quantity 
\begin{align*}
\norm{M}_\cb 
	&= \min_{M = L R} \set{\norm{L}_{2 \to \infty}\norm{R}_{1 \to 2}}  =  \max_{W} { \frac{\norm{W \bullet M}}{\norm{W}}}.
\end{align*}
is  the {\em completely bounded norm}\footnote{\citet{paulsen1986completely} in Section 7.7 attributes the second equality to \citet{haagerup1980decomposition}.} (abbreviated as cb-norm).
Here ${W \bullet M}$ denotes the Schur product~\cite{schur1911bemerkungen}. 

The factor $C_{\epsilon,\delta}\Delta_{\mathcal Q}\sqrt{\ln (6|\mathcal Q|}$ in equation~(\ref{eq:error_bound}) is  due to the error bound of the Gaussian mechanism followed by the union bound and is the same for all factorizations. Thus
to get a concrete additive error, we need to find a factorization $M=LR$ such that the quantity $\norm{M}_\cb$ is not just small, but also can be computed concretely. 
\emph{
Furthermore we observe that if both $L$ and $R$ are lower-triangular matrices, then the resulting mechanism works not only in the static setting, but also in the continual release model}.
Therefore, for the rest of the paper, we only focus on finding such a factorization of the workload matrix $\counting$, which is  a fundamental query  in the continual release model. 


\subsection{Our Results}
\label{sec:results}
\noindent \textbf{1. Bounding  $\norm{\counting}_\cb$.} The question of finding the optimal value of $\norm{\counting}_\cb$ was also raised in the conference version of \citet{matouvsek2020factorization}, who showed asymptotically tight bound. In the past, there has been considerable effort to get a tight bound on $\norm{\counting}_\cb$~\cite{davidson1984similarity, mathias1993hadamard} with the best-known result is the following due to Mathias~\cite[Corollary 3.5]{mathias1993hadamard}:
\begin{align}
\label{eq:mathias} 
\begin{split}
 \paren{\frac{1}{2} + \frac{1}{2T}} \widehat \gamma(T) \leq \norm{\counting}_\cb \leq \frac{\widehat \gamma(T) }{2} + \frac{1}{2}, 
\end{split}
\end{align}
$\text{where}~ \widehat \gamma(T) = \frac{1}{T} \sum_{j=1}^T \left\vert  {\csc \paren{\frac{(2j-1)\pi}{2T}}} \right\vert.$

The key point to note is that the proof of~\citet{mathias1993hadamard} relies on the dual characterization of cb-norm, and, thus, does not give an explicit factorization. 
In contrast, we give an explicit factorization into lower triangular matrices that achieve a bound in terms of the function, $\Psi:\mathbb N \to \mathbb R$ defined over natural numbers as follows:
{
\begin{align}
\Psi(T):=
1 + {1\over \pi}\paren{\ln\left(\dfrac{4T-3}{5}\right)}. 
\label{eq:psiT}
\end{align}
}
\begin{theorem}
[Upper bound on $\norm{\counting}_\cb$]
\label{thm:cb_bound}
Let $\counting \in \set{0,1}^{T \times T}$ be the matrix defined in \cref{eq:matrices}. Then, there is an explicit factorization $\counting = LR$ into lower triangular matrices such that
\begin{align}
    \norm{L}_{2 \to \infty} \norm{R}_{1 \to 2}
    \leq \Psi(T).\label{eq:counting}
\end{align}
\end{theorem}


\begin{figure}[t]
\vspace{-3mm}
\centering
\caption{\small{Difference between our upper bound and the explicitly computed Mathias lower bound.}}
\includegraphics[scale=0.5]{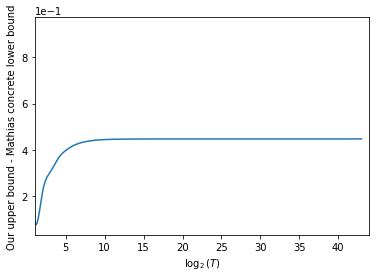}
    \label{fig:mathias_lower}
        \label{fig:mathias_upper}
\end{figure}





The bounds in \cref{eq:mathias} do not have a closed form formula; however, we can show that it converges in limit as $T \to \infty$ (\Cref{lem:gamma_limit}). However, our focus in this paper is on concrete bounds {and exact factorization}. 
Our (theoretical) upper bound and the (analytically computed) \emph{lower} bound of~\citet{mathias1993hadamard} are {less than $0.5$ apart for all $2^5 \leq T \leq 2^{44}$  (\Cref{fig:mathias_upper}).}

Additionally, our result has the advantage that we achieve the bound with an
explicit factorization of $\counting = L R$ such that both $L$ and $R$ are \emph{lower-triangular} matrices. As discussed above, this allows us to use it 
for various applications.  
{Using this fact and carefully choosing the ``noise vector'' for every time epoch, the following result is a consequence of \Cref{thm:cb_bound} and \cref{eq:error_bound}:
{
\begin{theorem}
[Upper bound on differentially private continual counting]
\label{thm:counting}
Let $\epsilon,\delta \in (0,1)$ be the privacy parameters. There is an $(\epsilon,\delta)$-differentially private  algorithm for binary counting in the  continual release model with output $a_t$ in round $t$ such that
in every execution, with probability at least $1-\beta$ over the coin tosses of the algorithm, 
simultaneously, for all rounds $t$ with $1 \le t \le T$, it holds that
\begin{align}
   \left\vert a_t - \sum_{i=1}^t x_i \right\vert \leq   C_{\epsilon,\delta} \Psi(t) \sqrt{2\ln(T/\beta)}, 
   \label{eq:binary_counting} 
\end{align}
where $\Psi(t)$ is as in \cref{eq:psiT}.
The time to output $a_t$ in round $t$ is $O(t)$. 
\end{theorem}
}

\noindent As mentioned above, the binary mechanism of~\citet{chan2011private} and \citet{dwork2010differentially} and its improvement (when the error metric is expected mean squared) by \citet{honaker2015efficient} can be seen as factorization mechanism. This has been independently noticed by \citet{mcmahan2022private}. While  \citet{chan2011private} and \citet{dwork2010differentially} do work in the continual release model,~\citet{honaker2015efficient}'s optimization does not  because for a partial sum, $\sum_{i\leq t} x_i$, it also uses the information stored at the nodes formed after time $t$. Therefore, for the comparison with related work, 
we do not include Honaker's optimization~\cite{honaker2015efficient}. Moreover, Honaker's optimization is for minimizing the expected mean squared error (i.e., in $\ell_2^2$ norm), and \emph{not} the $\ell_\infty$ error, which is the focus of our work. 
To the best of our knowledge, only \citet{chan2011private} and \citet{dwork2010differentially} consider additive error in terms of $\ell_\infty$ norm. 
All other approaches used for binary counting under continual observation (see~\cite{wang2021continuous} and references therein) use some form of smoothening of the output and consider the \emph{expected mean squared error}.  While useful in some applications,  many applications require to bound the additive $\ell_\infty$  error.  

The most relevant work to ours is the subsequent work by \citet{henzinger2022almost} that show that our factorization gives almost tight concrete bounds on performing counting under continual observation with respect to the \emph{expected mean squared error}. On the other hand, we bound the \emph{maximum absolute additive error}, i.e., in $\ell_\infty$ norm of the error. Note that our explicit factorization for $\counting$ has the nice property that there are exactly $T$ distinct entries (instead of possibly $T^2$ entries in \citet{mcmahan2022private})'s factorization. This has a large impact on computation time. 


\begin{remark}
[Suboptimality of the binary mechanism with respect to the constant]
\label{remark:suboptimality}
The binary mechanism computes a linear combination of the entries in the streamed vector $x$ as all of the internal nodes of the binary tree can be seen as a linear combinations of the entries of $x$. 
Now we can consider the binary mechanism as a matrix mechanism. The right factor $R_{\mathsf{binary}}$ is constructed as follows: $R_{\mathsf{binary}} = W_m$ where $W_1, \cdots, W_m$ are defined recursively as follows:
\begin{align*}
W_1 = \begin{pmatrix}
1
\end{pmatrix}, 
\quad W_k = \begin{pmatrix}
W_{k-1} & 0 \\
0 & W_{k-1} \\
1_{2^{k-2}} & 1_{2^{k-2}} 
\end{pmatrix}, \quad k \leq m.
\end{align*}

That is, $R_{\mathsf{binary}} = W_m \in \set{0,1}^{(2T  -1) \times T}$,
with each row corresponding to the partial sum computed for the corresponding node in the binary tree from leaf $i$ to the root.  
The corresponding matrix $L_{\mathsf{binary}}$ is a matrix in
$\set{0,1}^{T  \times (2T -1)}$,
where row $i$ 
has a one in at most $\log_2 (i)$ entries, corresponding exactly to the binary representation of $i$.

In particular, this factorization view tells us that $\norm{L}_{2 \to \infty} \norm{R}_{1 \to 2} = {\log_2(T)}{(\log_2(T)+1)}$, which 
 implies that the $\ell_\infty$-error
 is suboptimal by a factor of $\pi\log_2(e)$. Our experiments (\Cref{fig:running_count}) confirm this behavior experimentally. With respect to the \emph{ mean-squared error} this was observed by several works~\cite{mcmahan2022private, honaker2015efficient} culminating in the work of \citet{henzinger2022almost}, who gave a theoretical proof for the suboptimality of binary mechanism for mean-squared error. 
\end{remark}



\paragraph{Applications.}
Our result for continual binary counting can be extended in various directions. We show how to use it to quantify the additive error for (1) outputting a synthetic graph on the same vertex set which approximately preserves the values of all $(S,P)$-cuts with $S$ and $P$ being disjoint vertex sets of the  graph, (2) frequency estimation (aka histogram estimation), (3) various graph functions, (4) substring counting, and (5) episode counting. Our mechanism can also be adapted for the locally private non-interactive learners of~\citet{smith2017interaction} by replacing the binary mechanism with our matrix mechanism, 
{ which requires major changes in the analysis}. In \Cref{tab:results}, we tabulate these applications.
Based on a lower bound construction of~\citet{jain2021price}, we show in \Cref{sec:lower} that for large enough $T$ and constant $|S|$ the additive error in (1) is tight  up to polylogarithmic factors and the additive error in (4)  is tight for large enough $T$ up to a factor that is linear in $\log \log |{\mathcal U|} \ln T$, where $\mathcal{U}$ is the universe of letters (see \Cref{cor:substring} and \Cref{sec:lower} for details). {Finally, we can use our mechanism to estimate the running average of a sequence of $T$ bits with absolute error $\Psi(t) C_{\epsilon,\delta} \sqrt{\ln(6T)}/t$.} All the applications are presented in detail in \Cref{sec:applications}. We note that all our algorithms are differentially private in the  setting considered in \citet{dwork2010differentially}.

\begin{table*}[t]
\centering
\small{
\caption{Applications of \Cref{thm:cb_bound} ($\epsilon,\delta \in (0,1)$ are privacy parameter, $\eta \in (0,1)$ is the multiplicative approximation parameter,  $u$ is the dimension of each data item for histogram estimation and $b$ a bound on its $\ell_0$-sensitivity,  and $U$ is the set of letters, $\ell$ is the maximum length of the substrings that are counted, $T$ is the length of the stream), $n$ is the number of users in local-DP. Here, graph functions include subgraph counting, cost of minimum spanning tree, etc.}
   \begin{tabular}{lll} 
    \toprule
    Problem & Additive error & Reference \\ \midrule
    $(S,P)$-cuts (with $|S| \leq |T|$) & $3C_{\epsilon,\delta}|S| \Psi(T) \sqrt{(|S|+|P|)\ln (|S|+|P|) \ln(6T)  }$ & \Cref{cor:STcut} \\
    Histogram estimation & $C_{\epsilon,\delta}   \Psi(T)  \sqrt{ b \ln(6 \lvert U \rvert T) }$ & \Cref{cor:histogram} \\
    Graph functions & $C_{\epsilon,\delta} {\Psi(T) \sqrt{ \ln (6T)}}$ & \Cref{cor:graphFunction} \\
    Counting all length $\le \ell$ substrings & $C_{\epsilon,\delta} \Psi(T) \ell \sqrt{\ln (6T \lvert U \rvert^\ell)}$ & \Cref{cor:substring} \\
    Counting all length $\le \ell$ episodes & $2C_{\epsilon,\delta} \Psi(T) \ell \sqrt{\lvert U \rvert^{\ell-1} \ell \ln (6T \lvert U \rvert^\ell)}$ & \Cref{cor:episode} \\
    $1$-dimensional local convex risk min. & $\uniformBound + \frac{2}{\epsilon\sqrt{n}}$ & \Cref{cor:local-learning} \\
    \bottomrule
    \end{tabular}
}
\label{tab:results}
\end{table*}

\medskip
\noindent \textbf{2. Lower Bounds.}
We now turn our attention to new lower bounds on continual counting under approximate differential privacy. Prior to this paper, the only known lower bound for differentially private continual observation was for counting under pure differential privacy. There are a few other methods to prove lower bounds. For example, we can use the relation between hereditary discrepancy and private linear queries~\cite{muthukrishnan2012optimal} along with the generic lower bound on hereditary discrepancy~\cite{matouvsek2020factorization} to get an $\Omega(1)$ lower bound. This can be improved by using the reduction of continual counting to the threshold and get an $\Omega(\log^*(T))$~\cite{bun2015differentially}.

Our lower bound is for a class of mechanisms, called {\em data-independent mechanisms}, that add a random variable sampled from the distribution which is independent of the input. Most commonly used differentially private mechanisms, such as the Laplace mechanism and the Gaussian mechanism 
fall under this class.

For binary counting, recall from \cref{eq:mathias} that there exists a lower bound on $\norm{\counting}_\cb$ in terms of $\widehat \gamma(T)$. 
 {
In \Cref{lem:gamma_limit}, we show that 
\begin{align*}
\lim_{T \to \infty} \widehat \gamma(T) \to \frac{2\ln (T)}{\pi}
\end{align*}
from above. 
 Combined with the proof idea of \citet{edmonds2020power}, this gives the following bound for non-adaptive input streams, i.e., where the input stream is generated before seeing any output of the mechanism. 
 \begin{theorem}
 [Lower bound]
 \label{thm:lower_priv_counting}
 Let $\epsilon,\delta$ be a sufficiently small constant and $\mathfrak M$ be the set of data-independent $(\epsilon,\delta)$-differentially private algorithms for binary counting under non-adaptive continual observation. Then 
\[
\max_{x \in \set{0,1}^T} \mathbb E_{\mathcal M \in \mathfrak M} \sparen{ \norm{\mathcal M(x) - \counting x}_\infty^2 } \in  \Omega\paren{\frac{\ln^2(T)}{\epsilon^2}}.
\]
 \end{theorem}

That is, the variance of the $\ell_\infty$-error is $\Omega(\ln^2(T)/\epsilon^2)$. A formal proof of \Cref{thm:lower_priv_counting} is presented in~\Cref{sec:missingproofs}.

\subsection{Our Technical Contribution}
\noindent\textbf{1. Using the matrix mechanism in the continual release model.}
Our idea to use the matrix mechanism ${\cal F}$ in the  continual release model  is as follows: Assume $M$ is known
to ${\cal F}$ before any items of the stream $x$ arrive and there exists an explicit factorization of $M = L R$ into lower triangular matrices $L$ and $R$ that 
can be computed efficiently by ${\cal F}$ during preprocessing.
As we show, this is the case for the  matrix $\counting$.
This property leads to the following  mechanism: At time $t$,
the mechanism ${\cal F}$ creates $x'$ with consists of the current $x$-vector with $T-t$ zeros appended, 
and then returns the $t$\theth/ entry of $L(R x' + z)$, where $z$ is a suitable ``noise vector''.
As $L$ and $R$ are lower-triangular, the $t$\theth/ entry of $L(R x' + z)$ is identical to the $t$\theth/ entry in
$L(R x^f + z)$, where $x^f$ is the final input vector $x$, and, thus, it suffices to analyze the error of the static matrix mechanism. Note that our  algorithm can be implemented so that it requires time $O(t)$ at time $t$.
The advantage of this approach is  that it allows us to 
 use the analysis of the static
matrix mechanism while getting an \emph{explicit bound on the additive error} of the mechanism in the continual release model.

Factorization in terms of lower triangular matrices is not known to 
be necessary for the continual release model; however,  \citet{mcmahan2022private} pointed out that an arbitrary factorization will not work. For example, they discuss that \citet{honaker2015efficient}'s optimization of the binary mechanism can be seen as a factorization but it cannot be used for continual release as the output of his linear program at time $t$ can give {positive} weights to values of $x$ generated at a future time $t' > t$. 
Furthermore, as instead of computing $L(R x' + z)$ we work with $t \times t$-dimensional submatrices of $L$ and $R$, we can replace a $\log T$ factor in the upper bound  on the additive error due to the binary mechanism
by a $\ln(t)$, where $t \leq T$ denotes the current time step. 

\smallskip
\noindent\textbf{2. Bounding $\norm{\counting}_\cb$.}
The upper bound can be derived in various ways. One direct approach would be to find appropriate Kraus operators of a linear map and then use the characterization by \citet{haagerup1993bounded} of completely bounded norm. This approach yields an upper bound of $1+\frac{\ln (T)}{\pi}$; however, 
it does not directly give 
lower triangular factorization $L$ and $R$.

Instead we use the characterization given by \citet{paulsen1982completely}, which gives us a factorization in terms of lower triangular matrices. More precisely, using three basic trigonometric identities, we show that the  $(i,j)$\theth/ entry of $R$ and $L$ is an integral of every even power of the cosine function, $\frac{2}{\pi} \int\limits_{0}^{\pi/2} \cos^{2(i-j)}(\theta) \mathsf{d}\theta$\
for $i \geq j$. This choice of matrices leads to the upper bound in \cref{eq:counting}. 

\noindent \textbf{3. Applications.} While counting and averaging
under continual release follows from bounds in \Cref{thm:cb_bound}, computing cut-functions requires some ingenuity. In particular, one can consider $(S,P)$-cuts for an $n$-vertex graph $\cG =(V,E,w)$ as linear queries by constructing a matrix $M$ whose rows correspond to each possible cut query $(S,P) \in V \times V$ and whose columns corresponds to all possible edges in $\cG$. The entry $((S,P),j)$ of $M$ equals to $1$ if the edge $j$ crosses the boundary of the cut $(S,P)$. However, it is not clear how to use it in the matrix mechanism efficiently because known algorithm for finding a factorization as well as the resulting factorization depend polynomial on the dimension of the matrix and the number of rows in $M$ is $O(2^n)$. 
Instead we show how to exploit the algebraic structure of cut functions so that
at each time step $t$ the mechanism only has to compute
$L_t R(t) x(t)$, where $L_t$ is a ${n \choose 2} \times t {n \choose 2}$-dimensional matrix,
$R(t)$ is $t {n \choose 2} \times t {n \choose 2}$-dimensional matrix and $x(t)$ is $t {n \choose 2}$-dimensional. This gives an mechanism that has error
$O(\min(|S|, |P|) \ln (t) \sqrt{(|S| + |P|) \ln (|S| + |P|) \ln (6T)}$ (see \Cref{cor:STcut} for exact constant) and can be implemented to run in time $O(t n^4)$ per time step.

Binary counting can also be extended to histogram estimation.  \citet{cardoso2021differentially} 
gave a differentially private mechanism for histogram estimation where in each time epoch exactly one item either arrives or is removed. We extend our approach to work in the setting that they call \emph{known-domain, restricted $\ell_0$-sensitivity} and improve the constant factor in the additive error by the same amount as for binary counting.

We also show in \Cref{sec:local} an application of our mechanism to non-interactive local learning.  The non-interactive algorithm for local convex risk minimization is an adaption of the algorithm  of \citet{smith2017interaction}, which uses the binary tree mechanism for binary counting as a subroutine.  Replacing it by our mechanism  for binary counting (\Cref{thm:counting}) leads to various technical challenges:  From the algorithmic design perspective, \citet{smith2017interaction} used the binary mechanism with a randomization routine from  \citet{duchi2013local}, which expects as input a binary vector, while we apply randomization to  $Rx$, where $R$ has real-valued entries. We overcome this difficulty by using two instead of one binary counter. From an analysis point of view, the error analysis in \citet{smith2017interaction} is based on the error analysis in~\citet{bassily2015local} that uses various techniques, such as the randomizer of \citet{duchi2013local}, error-correcting codes, and the Johnson-Lindenstrauss lemma. However, one can show that we can give the same {\em uniform approximation} (see Definition 5 in~\citet{smith2017interaction} for the formal definition) as in \citet{smith2017interaction} by using the Gaussian mechanism and two binary counters} so that the rest of their analysis applies. 


\section{Notations and Preliminaries}
\label{sec:notation}
We use $v[i]$ to denote the $i$\theth/ coordinate of a vector $v$. For a matrix $A$, we use $A[i,j]$ to denote its $(i,j)$\theth/ coordinate, $A[:,i]$ to denote its $i$\theth/ column, $A[i,:]$ to denote its $i$\theth/ row, $\norm{A}_{\mathsf{tr}}$ to denote its trace norm of square matrix, $\Vert A \Vert_F$ to denote its Frobenius norm, $\norm{A}$ to denote its operator norm, and $A^\top$ to denote transpose of $A$. 
We use $\mathbb I_d$ to denote identity matrix of dimension $d$. If all the eigenvalues of a symmetrix matrix $S \in \R^{d \times d}$ are non-negative, then the matrix is known as {\it positive semidefinite}  (PSD for short) and is denoted by $S \succeq 0$. For symmetric matrices $A, B \in \R^{d\times d}$, the notations $A \preceq B$ implies that $B-A$ is PSD. {For an $a_1 \times a_2$ matrix $A$, its {\em tensor product} (or {\em Kronecker product}) with another matrix $B$ is 
\begin{align*}
\begin{pmatrix}
A[1,1] B & A[1,2]B &  \cdots & A[1,a_2]B \\
A[2,1] B & A[2,2]B &  \cdots & A[2,a_2]B \\
\vdots &  \ddots  & \vdots \\
A[a_1,1] B & A[a_1,2]B & \cdots & A[a_1,a_2]B \\
\end{pmatrix}. 
\end{align*}
We use $A \otimes B$ to denote the tensor product of $A$ and $B$.
}
In our case, the matrix $B$ would always be identity matrix of appropriate dimension.




One main application of our results is in {\em differential privacy} formally defined below: 
\begin{definition}
A randomized mechanism $\mathcal M$ gives $(\epsilon,\delta)$-differential privacy if for all {\em neighboring} data sets $D$ and $D'$ in the domain of $\mathcal M$ differing in at most one row, and all measurable subset $S$ in the range of $\mathcal M$, 
$
\pr \sparen{\mathcal M(D) \in S} \leq e^{-\eps} \pr \sparen{\mathcal M(D') \in S} + \delta,
$ 
where the probability is over the private coins of $\mathcal M$.
\end{definition}
This definition requires, however, to define \emph{neighboring} data sets in the continual release model. In this model the data is given as a \emph{stream} of individual data items, each belonging to a unique user, each arriving one after the other, one per time step.
In the privacy literature, there are two well-studied notions of neighboring streams~\cite{chan2011private, dwork2010differentially}: (i) {\em user-level privacy}, where two streams are neighboring if they differ in potentially all data items of a single user; and (ii) {\em event-level privacy}, where two streams are neighboring if they differ in a single data item  in the stream.
We here study event-level privacy.

Our algorithm uses the Gaussian mechanism. To define the Gaussian mechanism, we need to first define {\em $\ell_2$-sensitivity}. For a function $f : \mathcal X^n \to \R^d$  its $\ell_2$-sensitivity is defined as 
\begin{align}
\Delta f := \max_{\text{neighboring }X,X' \in \cX^n} \norm{f(X) - f(X')}_2.
\label{eq:ell_2sensitivity}    
\end{align}

\begin{definition}
[Gaussian mechanism]
\label{def:gaussian}
Let $f : \mathcal X^n \to \R^d$ be a function with $\ell_2$-sensitivity $\Delta f$. For a given $\epsilon,\delta \in (0,1)$ the Gaussian mechanism $\mathcal M$, which given $X \in \mathcal X^n$ returns $\mathcal M(X) =  f(X) + e$, where $e \sim \cN(0,C_{\epsilon,\delta}^2 (\Delta f)^2 \mathbb I_d)$, satisfies $(\epsilon,\delta)$-differential privacy.
\end{definition}

{
\begin{definition}
[Accuracy]
\label{def:accuracy}
A mechanism ${\cal M}$ is \emph{$(\alpha, T)$-accurate} for a function $f$ if, for all finite input streams $x$ of length $T$, the maximum absolute
error satisfies $||f(x) - {\cal M}(x)||_{\infty} \leq \alpha$ with probability at least $2/3$ over the coins of the mechanism.
\end{definition}

\section{Proof of Theorem~\ref{thm:cb_bound}}
\label{sec:proof}
The proof of Theorem~\ref{thm:cb_bound} relies on the following lemmas. 
\begin{lemma}
[\citet{chen2005best}]
\label{lem:s_m}
For integer $m$, define
$\mathcal{S}_m := \paren{\frac{1}{2}} \paren{\frac{3}{4}} \cdots \paren{\frac{2m-1}{2m}}.$ \text{ Then,  } 
$
\mathcal{S}_m  \leq \sqrt{\frac{1}{\pi (m+1/4)}}.
$
\end{lemma}

\begin{proof}
[Proof of  \Cref{thm:cb_bound}] 
Define a function, $f: \mathbb Z_+ \to \mathbb R$, recursively as follows
\begin{align}
f(0)=1 ~ \text{and} ~ f(k) =  \paren{\frac{2k-1}{2k}} f(k-1) ; \; k \geq 1 . 
\label{eq:functionF}
\end{align}

Since the function $f$ satisfies a nice recurrence relation, it is very easy to compute. Let $L$ and $R$ be defined as follows\footnote{Recently, Amir Yehudayoff (through Rasmus Pagh) communicated to us that this factorization was stated in the 1977 work by Bennett~\cite[page 630]{bennett1977schur}}: 
\begin{align}
    R[i,j] = L[i,j] = f(i-j). \label{eq:LRmatrix}
\end{align}



\begin{lemma}\label{lem:fac}
Let $\counting \in \set{0,1}^{T \times T}$ be the matrix defined in \cref{eq:matrices}. Then
$\counting = L R$.
\end{lemma}
\begin{proof}

    One way to prove the lemma is by using trigonometric inequalities for $\cos(2 \theta)$ and 
     that 
(i) For any $\theta \in [-\pi, \pi]$, $\sin^2(\theta) + \cos^2(\theta) =1$, 
(ii) For even $m$, $\frac{2}{\pi} \int\limits_{0}^{\pi/2} \cos^m(\theta) \mathsf{d}\theta = \paren{\frac{1}{2}} \paren{\frac{3}{4}} \cdots \paren{\frac{m-1}{m}}$,
(iii) For all $\theta \in [-\pi, \pi]$, $\cos(2\theta) =  2 \cos^2(\theta) - 1.$ 
This proof would require a lot of algebraic manipulations. 
However, our proof relies on an observation that all the three matrices, $L, R,$ and $\counting$ are $T \times T$ principal submatrix of  Toeplitz operators. The main idea would be to represent Toeplitz operator in its functional form: Let $a_1, a_2, \cdots, \in \mathbb C$ denote the entries of the Toeplitz operator, $\mathcal T$, with complex entries. Then its unique associated symbol is 
    \[
    f_{\mathcal T}(z) = \sum_{n=0}^\infty a_n z^n,
    \]
    where $z \in \mathbb C$ such that $|z|=1$. We can then write $z =e^{\iota \theta}$ for $0 \leq \theta \leq 2\pi$. Then for operator ${\mathcal M}$ whose principal $T \times T$ submatrix is the matrix $\counting$, we arrive that its associated symbol is  
    \[
    f_{\mathcal M}(\theta) = \sum_{n=0}^\infty e^{\iota \theta n} = \paren{1-e^{\iota \theta}}^{-1}.
    \]
    
    Let $\mathcal L$ denote the Toeplitz operator whose $T \times T$ submatrix is the matrix $L$ and $\mathcal R$ denote the Toeplitz operator whose $T \times T$ submatrix is the matrix $R$. Then the symbol associated with $\mathcal L$ and $\mathcal R$ is 
    \begin{align}
        f_{{\mathcal R}}(\theta) = f_{{\mathcal L}}(\theta) &= \sum_{n=0}^\infty f(n) e^{\iota  \theta n} = 1 + \frac{1}{2}e^{\iota \theta} + \frac{3}{8}e^{2\iota \theta} + \cdots 
        \label{eq:functionL}
    \end{align}
    
    Now  
    $
    (1-x)^{-1/2} = 1 + \frac{1}{2} x + \frac{3}{8}x^2 + \cdots
    $. Therefore, comparing the terms, we can rewrite  \cref{eq:functionL} as follows:
    \begin{align*}
       f_{{\mathcal L}}(\theta) = \paren{1-e^{\iota \theta}}^{-1/2} \quad \text{and} \quad
       f_{{\mathcal R}}(\theta) = \paren{1-e^{\iota \theta}}^{-1/2}.
    \end{align*}

   Since, for any two Toeplitz operators, $\mathcal A$ and $\mathcal B$ with associated symbols $f_{\mathcal A}(\theta)$ and $f_{\mathcal B}(\theta)$, respectively, $\mathcal A \mathcal B$ has the associated symbol $f_{\mathcal A}(\theta) f_{\mathcal B}(\theta)$~\cite{bottcher2000toeplitz}, we have the lemma. 
\end{proof}

Using Lemma~\ref{lem:s_m}, we have 
\begin{align*}
    \norm{L}_{2\to \infty}^2 & = \sum_{i=0}^{T-1} f(i)^2 = 1 + \sum_{i=1}^{T-1} f(i)^2 \\
    &\leq 1 + {1\over \pi}\sum_{i=1}^{T-1} {1 \over i+1/4} \\
    &\leq 1 + {1\over \pi}\int\limits_{i=1}^{T-1} {1 \over i+1/4} \\
    & \leq 1 + {1\over \pi}\paren{\ln\left(\dfrac{4T-3}{5}\right)}
\end{align*}

Combining with the fact that $\norm{L}_{2 \to \infty} = \norm{L}_{1 \to 2}$ and $R = L$, we have the following:
{
\begin{lemma}\label{lem:c}
Let $L$ and $R$
be $T \times T$ matrices defined by \cref{eq:LRmatrix}. Then $\norm{L}_{1 \to 2} =  \norm{L}_{2 \to \infty} = \norm{R}_{1 \to 2} = \norm{R}_{2 \to \infty} $. Further, 
 {$$ 
 \norm{L}_{2 \to \infty}^2 \leq 1 + {1\over \pi}\paren{\ln\left(\dfrac{4T-3}{5}\right)}. 
 $$}
\end{lemma}

}
Theorem \ref{thm:cb_bound} follows from Lemma~\ref{lem:fac} and \ref{lem:c}.
\end{proof}

\section{Proof of Theorem~\ref{thm:counting}}
Fix a time $t \leq T$. Let $L_t$ denote the $t\times t$ principal submatrix of $L$ and $R_t$ be the $t \times t$ principal submatrix of $R$. Let the vector formed by the streamed bits be $x_t = \begin{pmatrix}
x[1] & \cdots & x[t]
\end{pmatrix} \in \set{0,1}^t$. Let $z_t = \begin{pmatrix}
z[1] & \cdots & z[t]
\end{pmatrix}$ be a  freshly sampled Gaussian vector such that $z[i] \sim \cN(0,C_{\epsilon,\delta}^2 \norm{R_t}_{1 \to 2}^2)$. 

Let $\counting(t)$ denote the $t \times t$ principal submatrix of $\counting$. The algorithm computes
\begin{align*}
\widetilde x_t = L_t (R_t x_t + z_t) &= L_t R_t x_t + L_t z_t = \counting(t) x_t + L_t z_t  
\end{align*}
and outputs the $t$\theth/ coordinate of $\widetilde x_t$ (denoted by $x_t[t]$).
So far, the data structure maintained by the mechanism when it enters round $t$ is simply $x_{t-1}$. Now, note that
the $t$\theth/ coordinate of $\counting(t) x_t$ equals the sum of the first $t$ bits and can be computed in constant time when the sum of the bits of the previous round $t-1$ is known (i.e., maintained). Then, sampling a fresh Gaussian vector $z_t$ and computing the $t$\theth/ coordinate of $L_t z_t$ takes time $O(t)$.
{For privacy, note that the $\ell_2$-sensitivity of $R_tx_t$ is $\norm{R_t}_{1 \to 2}$; therefore, adding Gaussian noise with variance $\sigma_t^2 = C_{\epsilon,\delta}^2 \norm{R_t}_{1 \to 2}^2$ preserves $(\epsilon,\delta)$-differential privacy. Now for the accuracy guarantee,}
\[
\widetilde x_t[t] = \sum_{i=1}^t x[i] + \sum_{i=1}^t L_t[t,i] z_t[i].
\]

Therefore, 
\[
\left\vert \widetilde x_t[t] - \sum_{i=1}^t x[i] \right\vert = \left\vert \sum_{i=1}^t L_t[t,i] z_t[i] \right\vert.
\]

Recall that $z[i] \sim \cN(0,\sigma_t^2)$.
The Cauchy-Schwarz inequality shows that the function $\ell(z_t) := \sum_{i=1}^t L_t[t,i] z[i]$ has Lipschitz constant
$\norm{L_t}_{2 \to \infty}$, i.e.,~the maximum row norm.  Now define $z'[i] := z[i]/\sigma_t$
and note that $z'[i] \sim \cN(0,1)$ and
$ \mathbb{E}[\ell(z_t')]  = \mathbb{E}[\ell(z_t)] = 0$.
Using the concentration inequality for Gaussian random variables with unit variance and a function $f$ with Lipschitz constant $\norm{L_t}_{2 \to \infty}$ (e.g., Proposition 4 in~\citet{zeitouni2016gaussian}) implies that
 
\begin{align*}
&\pr_{z_t} \sparen{ \left\vert \ell(z_t) - 
\mathbb{E} [\ell(z_t)] \right \vert > a} \\
&= \pr_{z_t} \sparen{ \left\vert \ell(z_t') - 
\mathbb{E} [\ell(z_t')] \right \vert  > a/\sigma_t} \leq 2 e^{-a^2/(2 \sigma^2_t \norm{L_t}_{2 \to \infty}^2)}.
\end{align*}

Setting $$a = \sqrt{2}\sigma_t \norm{L_t}_{2 \to \infty} \log(2T/\beta) = \sqrt{2}C_{\epsilon,\delta}\norm{R_t}_{1\to 2}\norm{L_t}_{2 \to \infty} \log(2T/\beta),$$ the result follows using a union bound {over all $T$ time steps} and \Cref{thm:cb_bound},  
which implies that  
$\norm{L_t}_{2 \to \infty} = \norm{R_t}_{1 \to 2} \le \sqrt{\Psi(t)}
$.

 {
}

\section{Applications in Continual Release}
\label{sec:applications}

}

\subsection{Continuously releasing a synthetic graph which approximates  all cuts.}
\label{sec:app_cut}
For a weighted graph $\cG = (V, E, w)$, we let $n$ denote the size of the vertex set $V$ and $m$ denote the size of the edge set $E$. When the graph is uniformly weighted (i.e., all existing edges have the same weight, all non-existing have weight 0), then the graph is denoted $\cG = (V, E)$. 
Let $W$ be a diagonal matrix with non-negative edge weights on the diagonal. If we define an orientation of the edges of graph, then we can define the {\it signed edge-vertex incidence matrix} $A_{\mathcal{G}} \in \mathbb{R}^{m \times n}$ as follows:
\begin{align*}
A_{\mathcal{G}}[e, v] =  \left\{
\begin{array}{rl}
1 & \text{if } v \text{ is } e\text{'s head},\\
-1 & \text{if } v \text{ is } e\text{'s tail},\\
0 & \text{otherwise}.
\end{array} \right.
\end{align*}

One important matrix representation of a graph is its {\em Laplacian} (or {\em Kirchhoff matrix}). For a graph $\cG$, its Laplacian $K_\cG$ is the matrix form of the negative discrete Laplace operator on a graph that approximates the negative continuous Laplacian obtained by the finite difference method. 

\begin{definition}
[$(S,P)$-cut]
For two disjoint subsets $S$ and $P$, the size of the cut $(S,P)$-cut is denoted $\Phi_{S, P}(\mathcal{G})$ and defined as 
\begin{align*}
\Phi_{S,P}(\cG):= \sum_{u \in S, v \in P} w\paren{u,v}.
\end{align*}
When $P = V\backslash S$, we denote $\Phi_{S,P}(\cG)$ by $\Phi_S(\mathcal{G})$.
\end{definition}

In this section we study the following problem.
Let $\cG = (V,E,w)$ be a  weighted graph and consider a sequence of $T$ updates to the edges of $\cG$, where each update consists of an (edge,weight) tuple with weight in $[0,1]$ that adds the weight to the corresponding edge. For $t$, $1 \leq t \leq T$, let $\cG_t$ denote the graph that is obtained by applying the first $t$ updates to $\cG$. We give a differentially private mechanism that returns after each update $t$ a graph $\overline \cG_t$ such that for every cut $(S, P)$ with $S \cap P = \emptyset$, the number of edges crossing the cut in $\overline \cG_t$ differs from the number of edges crossing the same cut in $\cG_t$ by at most $O(\min\{|S|, |P|\} \sqrt{n \ln n} \ln^{3/2} T)$:


\begin{corollary}
\label{cor:STcut}
There is an  $(\epsilon,\delta)$-differentially private algorithm that, for any stream of edge updates of length $T >0$, 
outputs a synthetic graph $\overline \cG_t$ 
in round $t$ and
with probability at least 2/3 over the coin tosses of the algorithm, simultaneously for all rounds $t$ with $1 \le t \le T$,  it holds that  for any $S, P \subset V$ with $S \cap P = \emptyset$,
 
{
\begin{align*}
\Phi_{S,P} (\overline{\mathcal G}_t) &\leq \Phi_{S,P} (\mathcal{G}_t) +  3C_{\epsilon,\delta} |S|\psi(t) \sqrt{(|S|+|P|)\ln (|S|+|P|) \ln(6T)} ,
\end{align*}
}
where $\mathcal G_t$ is the graph formed at time $t$ through edge updates and $C_{\epsilon,\delta}$ is as defined in \cref{eq:error_bound}. 
The time for round $t$ is $O(t n^4)$.
\end{corollary}

\begin{proof}
Let us first analyze the case where $P = V \setminus S$.
In this case, we encode the updates as an $\mathbb{R}^{n \choose 2}$ vector and consider the following counting matrix:
\begin{align*}
\cut = \counting \otimes \mathbb I_{n \choose 2} 
\in \set{ 0,1}^{T {n \choose 2} \times T  {n \choose 2}}
\end{align*}
For the rest of this subsection, we drop the subscript and denote $\mathbb I_{n \choose 2}$ by $\mathbb I$. Recall the function $f$ defined by \cref{eq:functionF}. Let $L_{\mathsf{count}}[i,j] = f(i-j)$. Using this function, we can compute the following factorization of $\cut$:
$L = L_{\mathsf{count}} \otimes \mathbb I \text{ and } R=L.$ 
Let $R(t)$ and $L(t)$ denote the $t {n \choose 2} \times t {n \choose 2}$ principal submatrix of $R$ and $L$, respectively. { They both consist of $t$ \emph{block matrices}, each being formed by all columns and $n \choose 2$ rows of $R$ and $L$, respectively.  Further, let $R_t$ and $L_t$ denote the $t$\theth/ block matrix of $n \choose 2$ rows of $R$ and $L$, respectively. }
Let $x(t)$ be the $t \times {n \choose 2}$ vector formed by the first $t$ updates, i.e., the edges of $\cG_t$ which are given by the
${n \choose 2}$ vector 
$E_{\cG_t} := L_t R(t) x(t).$

Let $C_{\epsilon,\delta}$ be the function of $\epsilon$ and $\delta$ stated in \cref{eq:binary_counting} and $\sigma^2 = C_{\epsilon,\delta}^2 \norm{R_t}_{1 \to 2}^2 \sqrt{\ln (6T)}$. Then the edges of the {weighted} graph $\overline \cG_t$ which is output at time $t$ are given by the ${n \choose 2}$ vector 
$ L_t \left(R(t) x(t) + z\right),$ 
where $z \sim \mathcal N(0,\sigma^2)^{t \times {n \choose 2} }$. 
Note that computing the output naively takes time 
$O(t^2 n^4)$ to compute $R(t)x(t)$, time $O(tn^2)$ to generate and add $z$, and time $O(t n^4)$ to multiply the result with $L_t$. However, if we store the vector of $R(t-1)x(t-1)$ of the previous round and only
compute $R_t x(t)$ in round $t$, then the vector $R(t) x(t)$ can be created by  ``appending'' $R_t x(t)$ to the vector $R(t-1)x(t-1)$. Thus, $R(t)x(t)$
can be computed in time $O(t n^4)$, which reduces the total  computation time at time step $t$ to $O(t n^4)$.

We next analyse the additive error of this mechanism.
The output at time $t$ of the algorithm is a vector indicating the edges as 
$E_{\overline \cG_t} = L_t(R(t)x(t) + z) =  E_{\cG_t} + L_t z.    
$ 
 Let $f_t = \begin{pmatrix} f(t-1)  & \cdots & f(0) \end{pmatrix}^\top \in \R^t$ be a row vector whose coordinates are the evaluations of the function  $f(\cdot)$ on $\set{0,1,\cdots, t-1}$. Then 
$L_t = f_t \otimes \mathbb I$. 


As $L_t z$ is the weighted sum of random variables from $N(0,\sigma^2)$, it holds that 
$
L_t z \sim N(0,  \norm{L_t}^2_{2 \to \infty} \sigma^2 {\mathbb I}_{n \choose 2}).
$ In other words, the error is due to a graph, $\mathcal R$, with weights sampled from a Gaussian distribution $N(0, C_{\epsilon,\delta}^2 \norm{L_t}^2_{2 \to \infty} \norm{R(t)}_{1 \to 2}^2{\mathbb I})$. 

For a subset $S \subseteq [n]$, let
 {
\begin{align*}
\chi_S = \sum_{i \in S} \overline{e}_i, 
\end{align*}}
  
where $\overline{e}_i$ is the $i$\theth/ standard basis.
It is known that for any positive weighted graph ${\cG}$, the $(S,V \backslash S)$-cut  $\Phi_S(\cG) = \chi_S^\top K_{\cG} \chi_S$. 
So, for any subset $S \subseteq [n]$, 
$\vert \Phi_S(\overline \cG_t) -\Phi_S(\cG_t) \vert \leq \sqrt{\Psi(t)} \left\vert\chi_S^\top K_\cR \chi_S\right\vert.
$

The proof now follows on the same line as in Upadhyay {\it et al.}~\cite{upadhyay2021differentially}. In more details, if $K_n$ denotes the Laplacian of the complete graph with $n$ vertices, then 
$K_\cR \preceq 3\sigma   \sqrt{\frac{\ln (n)}{n}}  {K_n} $. Here, $A\preceq B$ means that $x^\top (B-A)x \geq 0$ for all $x$.
Setting $\sigma = 3 C_{\epsilon,\delta} \Psi(t)\sqrt{\ln (6T)}$,  the union bound gives that with probability at least $2/3$,
simultaneously for all rounds $T$, 
\begin{align}
\left\vert\chi_S^\top K_\cR \chi_S\right\vert &\leq  3\sigma   \sqrt{\frac{\ln (n)}{n}} \left\vert\chi_S^\top {L_n} \chi_S\right\vert  = {3\sigma  |S|\left(n - |S|\right)}\sqrt{\frac{\ln (n)}{n}}  \nonumber\\ &  \leq  3\sigma  \sqrt{n\ln (n)}|S| = 3C_{\epsilon,\delta}|S| \Psi(t) 
\sqrt{n\ln (n) \ln (6T)  }. \label{eq:cutR_complete}
\end{align}

We next consider the case of $(S,P)$ cuts, where $S \cup P \subseteq V$ and $S \cap P = \emptyset$. Without loss of generality, let $|S| \leq |P|$. Let us denote by $\mathcal G_A$ the graph induced by a vertex set $A \subseteq V$. In this case, for the analysis, we can consider the subgraph, $\mathcal G_{S\cup P}$, formed by the vertex set $S \cup P$. By Fiedler's result~\cite{fiedler1973algebraic}, $s_i(\mathcal G_{S\cup P}) \leq s_i(\mathcal G_{V})$, where  $s_i(\mathcal H)$ denotes the $i$\theth/ singular value of the Laplacian of the graph $\mathcal H$. Considering this subgraph, we have reduced the analysis of $(S,P)$ cut on $\mathcal G$ to the analysis of $(S, \overline S)$-cut on $\mathcal G_{S \cup P}$. Therefore, using the previous analysis, we get the result.

We now give the privacy proof. At time $t$, we only need to prove privacy for $R(t)x(t) + z$ as multiplication by $L_t$ can be seen as post-processing. Now consider two neighboring graphs form by the stream of (edge, weight) tuple, where at time $\tau$, they differ in weight. If $t < \tau$, the output distribution is the same because the input is the same. So, consider $t \geq \tau$. At this point, $x(t)$ and $x'(t)$ differ in exactly $t-\tau$ positions  by at most $1$. Breaking $x(t)-x'(t)$ in blocks of $n \choose 2$ coordinates, the position where $x(t)-x'(t)$ is $1$ is exactly corresponding to the edge they differ. Now multiplying with $R(t)$ results in vector whose non-zero entries are $\set{f(\tau), \cdots, f(t)}$. Using \Cref{lem:c}, $\norm{R(t)(x(t)-x'(t))}_2^2 = \norm{R_t}_2^2 - \norm{R_\tau}_2^2 \leq \norm{R_t}_2^2$. Therefore, we have $(\epsilon,\delta)$-differential privacy  from the choice of $\sigma$ and \Cref{def:gaussian}.
\end{proof}

\begin{remark}
In the worst case when $|S| = cn$ for some constant $c >0$, this results in an additive error of order $n^{3/2}\sqrt{\ln n} \ln^{3/2}T$. This result gives a mechanism for maintaining the minimum cut as well as a mechanism for maintaining the maximum cut, sparsest cuts, etc with such an additive error. Moreover, we can extend the result to receive updates with weights in $[-1,1]$ as long as the underlying graph only has positive weights at all time.
\end{remark}

For maintaining the minimum cut in the continual release model we show in
Appendix~\ref{sec:lower} that our upper bound is tight  up to polylogarithmic factors in $n$ and $T$ for large enough $T$ and constant $S$ using a reduction from a lower bound in~\cite{jain2021price}.

Our mechanism can implement a mechanism for the static setting as it allows to insert all edges of the static graph in one time step. The additive error that we achieve can  even be a slight improvement
over the additive error of $O(\sqrt{nm/\epsilon} \ln^2(n/\delta))$,
where $m$ is the sum of the edge weights, achieved by the mechanism in~\cite{eliavs2020differentially}. Note also that our bound does not contradict the lower bound for the additive error in that paper, as they show a lower bound only for the case that $\max\{|S|, |P|\} = \Omega(n)$.

\medskip
\subsection{Continual histogram}
\medskip
\noindent \textbf{Continual histogram.}
Modifying the analysis for cut functions, we can use our algorithm to compute the histogram of each column for a data base of $u$-dimensional binary vectors in the continual release model in a very straightforward manner.
Said differently, assume $\mathcal U$ is a universe of size $u$ and the input at a time step consists of the indicator vector of a subset of $\mathcal U$, which is a $u$-dimensional binary vector.
Let $b$ with $1 \le b \le u$ be the maximum number of 1s in the vector, i.e., the maximum size of a subset given at a time step.

\begin{corollary}
\label{cor:histogram}
Let $\mathcal U$ be the universe of size $u$ and let $1 \le b \le u$ be a given integer.
Consider a stream of $T$ vectors $x_t\in \set{0,1}^u$ such that $x_t[j]=1$ if $j \in \mathcal S$ and $x_t[k] = 0$ for all $k \not\in \mathcal S$ where at time $t$ the subset ${\mathcal S} \subseteq { \mathcal U}$ with $|{\mathcal S}| \le b$ is streamed. 
Then there is an efficient $(\epsilon,\delta)$-differentially private algorithm which outputs a vector $h_t\in \R^{u}$ in round $t$ such that, with probability at least 2/3 over the coin tosses of the algorithm, simultaneously, for all rounds $t$ with $1 \le t \le T$, it holds that
{
\begin{align*}
\norm{h_t - \sum_{i=1}^t x_i}_\infty \leq C_{\epsilon,\delta}   \Psi(t) \sqrt{ b \ln (6uT) }.
\label{eq:histogram}    
\end{align*}
}
The same bounds hold if items can also be removed, i.e., $x_t \in \set{-1,0,1}^u$ 
as long as $\sum_{i=1}^t x_i[j] \ge 0$ for all $1\leq j\leq u$ and all $t\leq T$. 
\end{corollary}


\begin{proof}
 We consider the following matrix:
 $
 \hist = \counting \otimes \mathbb I_{u}
 $
with every update being the indicator vector in $\set{0,1}^u$. We drop the subscript on $\mathbb I$ and denote $\mathbb I_u$ by $\mathbb I$ in the remainder of this subsection. 
Recall the function $f$ defined by \cref{eq:functionF}. Let $L_{\mathsf{count}}[i,j] = f(i-j)$. Using this function, we can compute the following factorization of $\hist$:
$L = L_{\mathsf{count}} \otimes \mathbb I \text{ and } R=L.$  
Let $R(t)$, resp.~$L(t)$, be the $tu \times tu$ principal submatrix of $R$, resp.~$L$.
{Let $R_t$, resp. $L_t$, be the $t$\theth/ block matrix of $R(t)$, resp.~$L(t)$, consisting of all columns and the last $u$ rows.}
Then at any time epoch we output $h_t = L_t (R(t)x(t) + z_t)$, where $x(t) \in \set{0,1}^{tu}$ is the row-wise stacking of $x_1,\cdots, x_t$ and $z_t[i] \sim N(0,\sigma^2_t)$ for $\sigma^2_t = C_{\epsilon,\delta}^2 b \norm{R_t}_{1 \to 2}^2 $.
For privacy, note that the $\ell_2$-sensitivity of $R_tx_t$ is $\sqrt{b} \norm{R_t}_{1 \to 2}$; therefore, adding Gaussian noise with variance $\sigma_t^2 = C_{\epsilon,\delta}^2 b \norm{R_t}_{1 \to 2}^2$ preserves $(\epsilon,\delta)$-differential privacy.

Using the same proof as in the case of $\counting$ we obtain 
\[
\norm{h_t - \sum_{i=1}^t x_i}_\infty \leq   C_{\epsilon,\delta}  \norm{R_T}_{1 \to 2} \norm{L_T}_{2 \to \infty}  \sqrt{b \ln (6u T)}.
\]

Using \Cref{thm:cb_bound}, we have the corollary.   
\end{proof}

\subsection{Other graph functions} 
Our upper bounds can also be applied to continual release algorithms that use the binary mechanism to compute prefix sums.
Let $f_1, f_2, \dots, f_T$ be a sequence $\sigma$ of $T$ function values.
The \emph{difference sequence of $\sigma$} is $f_2 - f_1, f_3 - f_2, \dots, f_T - f_{T-1}$.
For a graph function under continual release, its sensitivity may depend on the allowed types of edge updates. 
\citet{fichtenberger2021differentially} show that the $\ell_1$-sensitivity of the difference sequence of the cost of a minimum spanning tree, degree histograms, triangle count and $k$-star count does not depend on $T$ for partially dynamic updates (either edge insertions or deletions) and is $\Omega(T)$ for fully dynamic updates (edge insertions and/or deletions).
Using this result, they prove that one can privately compute these graph functions under continual observation by releasing noisy partial sums of the difference sequences of the respective functions.  
More generally, they show the following result for any graph function with bounded sensitivity of the difference sequence. 
\begin{lemma}[\citet{fichtenberger2021differentially}, cf Corollary 13]
    \label{lem:graph-obs1}
    Let $f$ be a graph function whose difference sequence has $\ell_1$-sensitivity $\Gamma$. Let $0 < p < 1$ and $\eps > 0$. For each $T \in \mathbb{N}$, the binary mechanism yields an $\epsilon$-differentially private algorithm to estimate $f$ on a graph sequence, which has additive error $O(\Gamma\eps^{-1}\cdot \ln^{3/2} T \cdot \ln p^{-1})$ with probability at least $1-p$.
\end{lemma}

We replace the summation by the binary mechanism in  {\Cref{lem:graph-obs1}} by summation using $\counting$, getting the following result.

{
\begin{corollary}
\label{cor:graphFunction}
        Let $f$ be a graph function whose difference sequence has $\ell_2$-sensitivity $\Gamma$.
        There is an $(\epsilon, \delta)$-differentially private algorithm that, for any sequence of edge updates of length $T > 0$ to a graph $\cG$, with probability at least 2/3 over the coin tosses of the algorithm, simultaneously for all rounds $t$ with $1 \le t \le T$, outputs at time $t$ an estimate of $f(\cG_t)$ that has additive error at most $C_{\epsilon,\delta} \Psi(t) \Gamma \sqrt{\ln(6T)}$. 
\end{corollary}
}

\subsection{Counting Substrings} 
We can also extend our mechanism  for counting all substrings of length at most $\ell$, where $\ell \geq 1$, in a sequence $\sigma$ of letters. After each update $i$ (i.e., a letter), we consider the binary vector $v_{\sigma,i}$ that is indexed by all substrings of length at most $\ell$. The value of $v_{\sigma,i}[s]$, which corresponds to the substring $s$, indicates whether the suffix of length $\lvert s \rvert$ of the current sequence equals $s$. We can cast the problem of counting substrings as a binary sum problem on the sequence of vectors $v_{\sigma,\cdot}$ and apply $\episode = \counting \otimes \mathbb{I}_u$ to the concatenated vectors, where $u = {\sum_{i \leq \ell} \lvert \mathcal{U} \rvert^i}$.

{
\begin{corollary}
\label{cor:substring}
    Let $\mathcal{U}$ be a universe of letters, let $\ell \geq 1$. There exists an $(\epsilon,\delta)$-differentially private algorithm that, given a sequence of letters $s = s_1 \cdots s_T$ from $\mathcal{U}$, outputs, after each letter, the approximate numbers of substrings of length at most $\ell$. With probability at least 2/3 over the coin tosses of the algorithm, simultaneously, for all rounds $t$ with $1 \le t \le T$, the algorithm has at time $t$ an additive error of at most $C_{\epsilon,\delta} \Psi(t) \ell \sqrt{\ln (2 \lvert U \rvert^\ell) \ln(6T)}$, where  $C_{\epsilon,\delta}$ is as defined in \cref{eq:error_bound}.
\end{corollary}
}

\begin{proof}
    Let $\sigma = \sigma_1 \cdots \sigma_T$ and $\sigma' = \sigma'_1 \cdots \sigma'_T$ be two sequences of letters that differ in only one position $p$, i.e., $\sigma_i = \sigma'_i$ if and only if $i \neq p$. We observe that $v_{\sigma,i} = v_{\sigma',i}$ for any $i \notin \{p, \ldots, p + \ell - 1\}$. Furthermore, for any $i$, $0 \leq i < \ell$ and $j$, $i+1 \leq j \leq \ell$, there exist only two substrings $s$ of length $j$ so that $v_{\sigma,p+i}[s] \neq v_{\sigma',p+i}[s]$. It follows that the $\ell_2$-sensitivity is at most $\sqrt{\sum_{i=0}^{\ell-1} \sum_{j=i+1}^{\ell} 2} \leq \sqrt{ \ell^2} = \ell$.
    {
    Using $\sum_{i \leq \ell} \lvert \mathcal{U} \rvert^i \leq 2 \lvert U \lvert^\ell$, the proof concludes analogously to the proof of \Cref{cor:histogram}.
    }
\end{proof}

\subsection{Episodes} 
Given a universe of \emph{events} (or \emph{letters}) $\mathcal{U}$, an \emph{episode} $e$ of length $\ell$ is a word over the alphabet $\mathcal{U}$, i.e., $e = e_1 \cdots e_\ell$ so that for each $i$, $1 \leq i \leq \ell$, $e_i \in \mathcal{U}$. Given a string $s = s_1 \cdots s_n \in \mathcal{U}^*$, an occurence of $e$ in $s$ is a subsequence of $s$ that equals $e$. A \emph{minimal} occurrence of an epsiode $e$ in $s$ is a subsequence of $s$ that equals $e$ and whose corresponding substring of $s$ does not contain another subsequence that equals $e$. In other words, $s_{i_1} \cdots s_{i_\ell}$ is a minimal occurence of $e$ in $s$ if and only if (1)
 for all $j$, $1 \leq j \leq \ell$, $s_{i_j} = e_j$ and
(2) there does not exist $s_{j_1} \cdots s_{j_\ell}$ so that for all $k$, $1 \leq k \leq \ell$, $s_{j_k} = e_k$, and either $i_1 < j_1$ and $j_\ell \leq i_\ell$, or $i_1 \leq j_1$ and $j_\ell < i_\ell$. 
The support of an episode $e$ on a string $s$ is the number of characters from the string that are part of at least one minimal occurrence of $e$. Note that for an episode $e$, its minimal occurrences may overlap.
For the non-differentially private setting, Lin {\it et al.}~\cite{lin2014frequent} provide an algorithm that dynamically maintains the number of minimal occurrences of episodes in a stream of events. For better performance, the counts may be restricted to those episodes that have some minimum support on the input (i.e., frequent episodes).

\begin{lemma}[\citet{lin2014frequent}]
    Let $\mathcal{U}$ be a universe of events, let $\ell \geq 2$, and let $S \ge 1$. There exists a (non-private) algorithm that, given a sequence of events $s = s_1 \cdots s_T$ from $\mathcal{U}$, outputs, after each event, the number of minimal occurrences for each episode of length at most $\ell$ that has support at least $S$. The time complexity per update is $\Tilde{O}(T / S + \lvert \mathcal{U} \rvert^2)$ and the space complexity of the algorithm is $\Tilde{O}(\lvert U \rvert \cdot T / S + \lvert \mathcal{U} \rvert^2 \cdot T)$.
\end{lemma}

There can be at most one minimal occurrence of $e$ that ends at a fixed element $s_t \in s$. Therefore, we can view the output of the algorithm after event $s_t$ as a binary vector $v_t \in \{0,1\}^{\sum_{i \leq \ell} \lvert \mathcal{U} \rvert^i}$ that is indexed by all episodes of length at most $\ell$ and that indicates, after each event $s_t$, if a minimal occurrences of epsiode $e$ ends at $s_t$. Summing up the (binary) entries corresponding to $e$ in $v_1, \ldots, v_t$ yields the number of minimal occurrences of $e$ in $s_1 \cdots s_t$. Therefore, we can cast this problem of counting minimal occurrences of episodes as a binary sum problem and apply $\episode$.

{
\begin{corollary}
\label{cor:episode}
    Let $\mathcal{U}$ be a universe of events, let $\ell \geq 2$, and let $S \ge 1$. There exists an $(\epsilon,\delta)$-differentially private  mechanism that, given a sequence of events $s = s_1 \cdots s_T$ from $\mathcal{U}$, outputs, after each event, the approximate number of minimal occurrences for each episode of length at most $\ell$ that has support at least $S$. With probability at least 2/3 over the coin tosses of the algorithm, simultaneously for all rounds $t$ with $1 \le t \le T$ the algorithm has at time $t$ an additive error of at most $C_{\epsilon,\delta} \Psi(t) \sqrt{\lvert U \rvert^{\ell-1} \ell \ln (2 \lvert U \rvert^\ell) \ln(6T)}$. 
\end{corollary}
}
\begin{proof}
    Let $\sigma = \sigma_1 \cdots \sigma_T$ and $\sigma' = \sigma'_1 \cdots \sigma'_T$ be two sequences of letters that differ in only one position $p$, i.e., $\sigma_i = \sigma'_i$ if and only if $i \neq p$. Recall that we are only interested in minimal occurences of episodes. Therefore, the number of query answers that are different for $\sigma$ and $\sigma'$ are trivially upper bounded by two times the maximum number of episodes that end on the same character (once for $\sigma[p]$ and once for $\sigma'[p]$), times the maximum length of an episode (as for every episode that ends at $p$, only the one with the latest start is a minimal occurrence). 
    This is bounded by $2 \sum_{i \leq \ell} \lvert \mathcal{U} \rvert^{i-1} \cdot \ell \leq 4 \lvert U \rvert ^{\ell-1} \ell$. 
    It follows that the global $\ell_2$-sensitivity is at most $2 \sqrt{\lvert U \rvert^{\ell-1} \ell}$.
    {
    Using $\sum_{i \leq \ell} \lvert \mathcal{U} \rvert^i \leq 2 \lvert U \lvert^\ell$, the proof concludes analogously to the proof of \Cref{cor:histogram}.
    }
\end{proof}

\noindent \textbf{On sparse streams.} 
Until now, we have focused mainly on worst case analysis. We next consider the case when the stream is sparse, i.e., number of ones in the stream is upper bounded by a parameter $s$. Dwork et al.~\cite{dwork2015pure} showed that under this assumption, one can asymtotically improve the error bound on continual release to $O\paren{\frac{\log(T) + \log^2(s)}{\epsilon}}$ while preserving $\epsilon$-differential privacy. Using their analysis combined with our bounds we directly get the error bound 
$C_{\epsilon,\delta} \paren{5\log(T) + \Psi(n)\sqrt{\ln(6n)}}.$

\subsection{Non-interactive Local Learning}
\label{sec:local}
In this section, we consider convex risk minimization in the non-interactive local differential privacy mode (LDP) using \Cref{thm:cb_bound}.  That is, there are $n$ participants (also known as {\em clients}) and one {\em server}. Every client has a private input $d_i$ from a fixed universe $\cD$. {To retain the privacy of this input, each client applies a differentially-private mechanism to their data  (\emph{local} model) and then sends a single message  to the server which allows the server to perform the desired computation (convex risk minimization in our case). After receiving all messages, the server outputs the result without further interaction with the clients (\emph{non-interactive}).}

In $1$-dimensional convex risk minimization, a problem is specified by a convex, closed and bounded constraint set $\cC$ in $\R$ and a function $\ell:\cC\times \cD \to \R$ which is convex in its first argument, that is,  for all $D\in \cD$, $\ell(\cdot;D)$ is convex. A data set $D = (d_1,\ldots,d_n) \in \mathcal D^n$ defines a loss (or \emph{empirical risk}) function: $\mathcal L(\theta;D)=\tfrac1 n \sum_{i=1}^n \ell(\theta; d_i)$, where $\theta$ is a variable that is chosen as to minimize the loss function. The goal of the algorithm is to output a function $f$ that assigns to each input $D$ a value $\theta \in \mathcal C$ that minimizes the average loss over the data sample $D$. 
For example, finding the median of the 1-dimensional data set $D \in [0,1]^n$ consisting of $n$ points in the interval $[0,1]$ corresponds to finding $\theta \in \mathcal C$ that minimizes the loss $\mathcal L(\theta,D) = \sum_i |\theta-d_i|$. 

When the inputs are drawn i.i.d. from an underlying distribution $\cP$ over the data universe $\cD$, one can also seek to minimize the \emph{population risk}: $\cL_\cP(\theta) = \E_{D\sim \cP}[\ell(\theta; D)].$ We will use some notations in this section. Let $I_1, \cdots, I_w$ be $w$ disjoint intervals of $[0,1]$ of size $s :=\floor{\frac{1}{\epsilon\sqrt{n}}}$. Let $\mathcal B = \set{j \cdot s: 0 \leq j \leq w}$. Given  a vector $a \in \R^w$ let $g$ be a ``continuous intrapolation'' of the vector $a$, namely
$g:[0,1]^w \times [0,1] \to [0,1]$ such that $ g(a,\theta) = a[k]$, where $k = \argmin_{z \in \mathcal B} | z-\theta |$,  with ties broken for smaller values.
Also, let
$f:\R^w \times [0,1] \to [0,1]$ be defined as  $f(a,x) = \int\limits_{0}^x  g(a,t) \mathsf{d}t$.

{
\citet{smith2017interaction}
showed the following:

\begin{theorem}
[Corollary 8 in~\citet{smith2017interaction}]
 \label{thm:approx}
 For every 1-Lipschitz\footnote{A function $\ell:\cC\to\R$, defined over $\cC$ endowed with $\ell_2$ norm, is $L$-Lipschitz with respect  if
 for all $\theta,\theta'\in \cC$, $|\ell(\theta) -
\ell(\theta')|\leq L \| \theta-\theta'\|_2.$} loss function $\ell:[0,1]\times \cD \to \R$,
  there is a randomized algorithm $Z: \cD \to [0,1]$,
  such that for
  every
  distribution $\cP$ on $\cD$, the  distribution $\cQ $ on
  $[0,1]$ obtained by running $Z$ on a single draw from $\cP$ satisfies
$\cL_\cP(\theta) = \Med_{\cQ}(\theta)$ {for all }$\theta
\in [0,1]$, where $\Med_{\cP}(\theta) = {\mathbb{E}}_{d\sim \cQ}[|\theta - d|]$.
\end{theorem}

In other words, differentially private small error for 1-dimensional median is enough to solve differentially private loss minimization for general $1$-Lipschitz functions. Prior work used a binary mechanism to determine the 1-dimensional median. We show how to replace this mechanism by the factorization mechanism of Theorem~\ref{thm:counting}. As the reduction in Theorem~\ref{thm:approx} preserves exactly the additive error, our analysis of the additive error in Theorem~\ref{thm:counting} carries through, giving a concrete upper bound on the additive error.

We first recall the  algorithm  of~\citet{smith2017interaction}. Median is non-differentiable at its minimizer $\theta^*$, but in any open interval around $\theta^*$, its gradient is either $+1$ or $-1$. 
$\stu$ first divides the interval $[0,1]$ into $w = \ceil{\epsilon\sqrt{n}}$ disjoint intervals $I_1, \cdots, I_w$ of $[0,1]$ of size $s :=\floor{\frac{1}{\epsilon\sqrt{n}}}$. Let $\mathcal B = \set{j \cdot s: 0 \leq j \leq w}$. Every client constructs a $w$-dimensional binary vector that has $1$ only on the coordinate $j$ if its data point $d_i \in I_j$. The client then executes the binary mechanism with the randomizer of Duchi {\it et al.}~\cite{duchi2013local} on its vector and sends the binary tree to the server.
Based on this information the server computes a vector $x^\stu \in \mathbb R^w$, where $x^\stu[j]$ is the
$1/n$ times the difference of the number of points in the interval
$\cup_{l=1}^j I_l$ and the number of data points in the interval $\cup_{l=j+1}^w I_l$.   The server
outputs the function $f(x^\stu,\theta)$.
}


To replace the binary tree mechanism used in \citet{smith2017interaction}  (dubbed as $\stu$) into a factorization mechanism based algorithm is not straightforward because of two reasons: (i) Smith {\it et al.} used the binary mechanism with a randomization routine from \citet{duchi2013local}, which expects as input a binary vector, while we apply randomization to  $Rx$, where $x$ is the binary vector, and (ii) the error analysis is based on the error analysis in~\citet{bassily2015local} which does not carry over to the factorization mechanism.

We now describe how we modify $\stu$ to give an LDP algorithm $\mathcal A$.
 Instead of forming a binary tree, 
every client $i$ forms two binary vectors $u_i, v_i \in \set{0,1}^{w}$ with $u_i[j] = v_i[w-j] = 1$  if $d_i \in I_j$ and 0 otherwise. Note that in both vectors exactly 1 bit is set and that
$$\left(\sum_{i=1}^n \sum_{l = 1}^t u_i[l]\right) - 
 \left(\sum_{i=1}^n \sum_{l = t+1}^w v_i[l] \right)$$ gives the number of bits in the interval
$\cup_{l=1}^t I_l$ minus the number of data points in the interval $\cup_{l=t+1}^w I_l$. The user now sends two vectors $y_i,z_i \in \R^{w}$ to the server formed by running the binary counter mechanism defined in \Cref{thm:counting} on $u_i$ and $v_i$ with privacy parameters $(\frac{\epsilon}{2}, \frac{\delta}{2})$. Since the client's message is computed using a differentially private mechanism for each vector, the resulting distributed mechanism is $(\epsilon,\delta)$-LDP using the basic composition theorem. 

On receiving these vectors, the server first computes the aggregate vector
\begin{align}
  \widehat x[t] = \frac{1}{n}  \left(\sum_{i=1}^n y_i[t] - \sum_{i=1}^n z_i[w-t] \right).
\label{eq:server_xyz}    
\end{align}

The server then computes and outputs $f(\widehat x, \theta)$. 


To analyze our mechanism let 
{$x^\stu$ be the vector formed in $\stu$
} and $\widetilde x$ 
be the vector that the server in $\stu$ would have formed if clients did not use any randomizer. Equation (3) in \citet{smith2017interaction} first showed that, for all 
\begin{align}
\label{eq:smithuniform}
\begin{split}
     \forall \theta \in [0,1], \quad  \left\vert g(x^\stu,\theta) - g(\widetilde x,\theta) \right\vert \leq \alpha \quad \\
     \text{for} \quad \alpha \in O \paren{\frac{\log^2(\eps^2
      n)\sqrt{\log(\eps^2n)}}{\varepsilon \sqrt{n}}}.   
\end{split}
\end{align}

\citet{smith2017interaction} (see Theorem 6) then use the fact that $f(x,\theta) = \int \limits_{0}^\theta g(x;s) \mathsf{d}s $   to show that $\left \vert f(x^\stu,\theta) - \Med_\cP(\theta) \right\vert \leq \left\vert g(x^\stu,\theta) - g(\widetilde x,\theta) \right\vert + \frac{2}{\epsilon{\sqrt{n}}}$ and use \cref{eq:smithuniform} to get their final bound, which is
$O \paren{\frac{\log^2(\eps^2
      n)\sqrt{\log(\eps^2n)}}{\varepsilon \sqrt{n}}}$.
We remark that we can replace $x^\stu$ by any $y \in \R^{w}$ as long as $\left\vert g(y,\theta) - g(\widetilde x,\theta) \right\vert \leq \alpha$ for all $\theta \in [0,1]$.

We now show an equivalent result to \cref{eq:smithuniform}. 
We argue that the vector $\widehat x$ serves the same purpose as $ x^\stu$.
The key observation here is that $\sum_{i=1}^n y_i[t]$ contains the partial sum for the intervals $I_1,\dots, I_t$ and $\sum_{i=1}^n z_i[w-t]$ contains the partial sum for $I_{j+1},\dots, I_w$. Let $\overline x = \frac{1}{n}  (\sum_{i=1}^n u_i[t] - \sum_{i=1}^n v_i[w-t])$ be the vector corresponding to the estimates in \cref{eq:server_xyz} if no privacy mechanism was used. Note that $\widetilde x = \overline x$. Since the randomness used by different clients is independent, 
{
\begin{align*}
&\mathsf{Var}[\widehat x[t]]= \frac{1}{{n^2}} \mathsf{Var}\sparen{ \sum_{i=1}^n (y_i[t] - z_i[w-t]) }\\
&= \frac{2}{n^2} \mathsf{Var}\sparen{ \sum_{i=1}^n y_i[t]} = \frac{2}{n^2} \sum_{i=1}^n  \mathsf{Var}\sparen{ y_i[t]} = \frac{2}{n} \sigma_t,
\end{align*}
}
where $\sigma_t$ is the variance used in the binary counting mechanism of Theorem~\ref{thm:counting}.
Using the concentration bound as in the proof of \Cref{thm:counting}, we have 
$\norm{\widehat x - \overline x}_\infty \leq  2\beta$ with $\beta = \uniformBound$. 
By the definition of $g(\cdot, \cdot)$, we therefore have $\forall \theta \in [0,1]$,
$\left\vert g(\widehat x, \theta ) - g(\overline x,\theta) \right \vert  \leq 2\beta.  $

Now using the same line of argument as in \citet{smith2017interaction}, we get the following bound:

\begin{corollary}
\label{cor:local-learning}
For every distribution
$\cP$ on $[0,1]$,  with probability $2/3$ over
 $D \sim\cP^n$ and  $\mathcal A$, the output $\widehat f \gets \mathcal A$ satisfies
$
 \left \vert f(\widehat x, \theta) -
  \Med_{\cP} (\theta) \right\vert  \leq 2 \beta + \frac{2}{\epsilon \sqrt{n}}$, where 
$\Med_{\cP}(\theta)={\mathbb{E}}_{d\sim \cQ}[|\theta - d|]$. Further, $\mathcal A$ is $(\epsilon,\delta)$-LDP.
\end{corollary}

Our algorithm $\mathcal A$ is non-interactive $(\epsilon,\delta)$-LDP algorithm and not $\epsilon$-LDP as $\stu$, but we can give $\mathcal A$ our algorithm as input to the {\sf GenProt} transformation (Algorithm 3) in~\citet{bun2018heavy} to turn it into a $(10\epsilon,0)$-LDP algorithm (see Lemma 6.2 in \citet{bun2018heavy}) at the cost of increasing the population risk using Theorem 6.1 in \citet{bun2018heavy}. 


\section{Lower bounds}\label{sec:lower}
\label{sec:missingproofs}

\begin{definition}
[{\scshape Max-Cut}]
\label{def:maxnonprivate}
Given a graph $\cG=(V, E, w)$, the {\it maximum cut} of the graph is the optimization problem
$$
\max_{S \subseteq V} \left\{\Phi_S(\cG)\right\} = \max_{S \subseteq V} \left\{\sum_{u \in S, v \in V\backslash S} w\paren{u,v} \right\}.
$$
Let $\mathsf{OPT}_{\mathsf {max}}(\cG)$ denote the maximum value. 
\end{definition}

In this section we use a reduction from the maximum sum problem.
Let ${\cal X} = \{0,1\}^d$, let $x \in {\cal X}^T$, $d \in \mathbb{N}$, and 
for $1 \leq j \leq d$, let $x_t[j]$ denote the $j$-th coordinate of record $x_t$.
A mechanism for the \emph{$d$-dimensional maximum sum problem under continual observation} is to return for each  $0 \leq t \leq T$, the value $\max_{1 \leq j \leq d} \sum_{s = 1}^{t} x_s[j]$.

In~\cite{jain2021price} Jain et al.~studied the problem of computing in the continual release model the
maximum sum of a $d$-dimensional vector. Two vectors $x$ and  $x'$ are neighboring if they differ in only one $d$-dimensional vectors $x_t$ for some $1 \leq t \leq T$. They showed that for any $(\epsilon, \delta)$-differentially private and $(\alpha,T)$-accurate mechanism for maximum sum problem under continual observation it holds that 
\begin{enumerate}
    \item $\alpha = \Omega\left(\min\{ \frac{T^{1/3}}{\epsilon^{2/3}\log^{2/3}(\epsilon T)},
    \frac{\sqrt d}{\epsilon \log d}, T\} \right)$ if $\delta > 0$ and $\delta = o(\epsilon/T)$;
    \item $\alpha = \Omega \left( \min \{ \sqrt{T/\epsilon}, d/\epsilon, T \} \right)$ if $\delta=0$.
\end{enumerate}

We use this fact to  show a lower bound for maintaining a minimum cut under continual observation, where each update consists of a set of edges that are inserted or deleted.

\begin{theorem}
For all $\epsilon \in (0,1), \delta \in [0,1),$ sufficiently large $T \in \mathbb{N}$ and any mechanism $\cal M$ that returns the value of the minimum cut in a multi-graph with at least 3 nodes in the continual release model, is $(\epsilon, \delta)$-differentially private, and
$(\alpha, T)$-accurate it holds that
\begin{enumerate}
    \item $\alpha = \Omega\left(\min\{ \frac{T^{1/3}}{\epsilon^{2/3}\log^{2/3}(\epsilon T)},
    \frac{\sqrt n}{\epsilon \log n}, T\} \right)$ if $\delta > 0$ and $\delta = o(\epsilon/T)$;
    \item $\alpha = \Omega \left( \min \{ \sqrt{\frac{T}{\epsilon}}, \frac{n}{\epsilon}, T \} \right)$ if $\delta=0$.
\end{enumerate}
The same hold for any mechanism maintaining the minimum degree.
\end{theorem}
\begin{proof}
Using a mechanism ${\cal M}$ for the minimum cut problem under continual observation for a graph ${\cal G}=(V,E)$ with $d +1$ nodes we show how to solve the $d$-dimensional maximum sum problem under continual observation.  During this reduction the input sequence of length $T$ for the maximum sum problem is transformed into a input sequence of length $T$ for the minimum cut problem. The lower bound then follows from this and the fact that $n = d+1$ in our reduction.

Let $\cal G$ be a clique with $d+1$ nodes such that one of the nodes is labeled $v$ and all other nodes are numbered consecutively by $1, \dots, d$. For every pair of nodes that does not contain $v$, give it $T$ parallel edges, and give every node $j$ with $1 \leq j \leq d$ $3T$ parallel edges to $v$. 
Note that $v$ has initially degree $3Td$, every other node has initially degree $T (d+2)$ and the minimum degree corresponds to the minimum cut. Whenever a new vector $x_t$ arrives, give to ${\cal M}$ a sequence update that removes one of the parallel edges $(v,j)$ for every $j$ with $x_t[j] = 1$. Let $j^*$ be the index that maximizes $\sum_{s = 1}^{t} x_s[j].$
Note that the corresponding node labeled $j^*$ has degree
$T (d+2) - \sum_{s = 1}^{t} x_s[j^*]$, while $v$ has degree at least $2Td \ge T (d+2)$ as $d+1 \ge 3$, and every other node has
degree at least $T (d+2) - \sum_{s = 1}^{t} x_s[j^*]$. 
Furthermore the minimum degree also gives the minimum cut in $\cal G$.
Thus ${\cal M}$ can be used to solve the maximum sum problem and the lower bound follows from the above.

Note that the proof also shows the result for a mechanism maintaining the minimum degree.
\end{proof}

It follows that for $T \geq n^{3/2}/\log n$ the additive error for any $(\epsilon, \delta)$-differentially private mechanism is $\Omega(\sqrt n/(\epsilon \log n))$, which implies that our additive error  is tight up to a factor of $\log n \log^{3/2} T$ if the minimum cut $S$ has constant size.

 We now show a lower bound for counting substrings: 

\begin{theorem}
For all $\epsilon \in (0,1), \delta \in [0,1),$ sufficiently large $T \in \mathbb{N}$, universe $\cal U$, $\ell \ge 1$ and $S \ge 1$ and for any mechanism $\cal M$ that, given a sequence $s$ of letters from $U$, outputs, after each letter the approximate number of substrings of length at most $\ell$ that has support at least $S$, is $(\epsilon, \delta)$-differentially private, and
$(\alpha, T)$-accurate it holds that
\begin{enumerate}
    \item $\alpha = \Omega\left(\min\{ \frac{T^{1/3}}{\epsilon^{2/3} \log^{2/3}(\epsilon T)},
    \frac{\sqrt{\log |U|}}{\epsilon \log \log |U|}  , T \} \right)$ if $\delta > 0$ and $\delta = o(\epsilon/T)$;
    \item $\alpha = \Omega \left( \min \{ \sqrt{T/\epsilon}, \log |U|/\epsilon, T \} \right)$ if $\delta=0$.
\end{enumerate}
\end{theorem}
\begin{proof}
Using a mechanism  for substring counting under continual observation 
up to length $\ell = 1$ and  universe $\cal U$ of letters of size $2^{d}$
we show how to create a mechanism $\cal M$ for the $d$-dimensional maximum sum problem under continual observation. During this reduction the input sequence of length $T$ for the maximum sum problem is transformed into a sequence of length $T$. The lower bound follows from this and the fact that $d = \log |U|$. 

Let $\cal U$ consist of $2^d $ many  letters $s_p$ for $1 \leq p \leq 2^d$, one per possible record in ${\cal X} = \{0,1\}^d$.
Given a $d$-dimensional bit-vector $x_t$ at time step $t$ we append to the input string $s$ the corresponding letter $|U|$. Thus, two neighboring inputs $x, x' \in {\cal X}^T$
for the maximum sum problem
lead to two neighboring sequences $s$ and $s'$ for the substring counting problem.
The substring counting mechanism outputs at time step $t$ an approximate count of all
substrings of length $1$ with maximum error $\alpha$ over all counts and all time steps. Our mechanism $\cal M$ determines the maximum count returned for any substring of length $1$ and returns it. This gives
an answer to the maximum sum problem with additive error at most $\alpha$. 
\end{proof}

This implies that for large enough $T$ and constant $\ell$ the additive error of our mechanism is tight up to a factor of $ \log \log |U| \log^{3/2} T$.

\begin{proof}
[Proof of~\Cref{thm:lower_priv_counting}] Define the function, $f(t):=\frac{2}{\pi} \log \paren{\cot \paren{\frac{\pi}{4t}} }.$ It is easy to see that 
$
f(t) = \frac{1}{t} \int \limits_{1}^t \left\vert {\csc \paren{\frac{(2x-1)\pi}{2t}}} \right\vert \mathsf dx.    
$
 From the basic approximation rule of Reimann integration, this implies that 
$ \widehat \gamma_t \leq f(t).$
{The following limit follows using L'Hospital rule:}
\begin{align*}
\lim_{t \to \infty} \frac{2}{\pi} \frac{\ln\paren{\cot\paren{\frac{\pi}{4t}}}}{\ln(t)} & = 
    \lim_{t \to \infty} \frac{\csc^2\left(\frac{{\pi}}{4t}\right)}{2t \cot\left(\frac{{\pi}}{4t}\right)} 
    = \frac{2}{\pi}.
\end{align*}
That is, we have the following:
\begin{lemma}
\label{lem:gamma_limit}
$\lim_{t \to \infty} \frac{1}{t} \sum_{j=1}^t \left\vert \{\csc \paren{\frac{(2j-1)\pi}{2t}} \right\vert 
\to \frac{2\ln(t)}{\pi}
$ from above. 
\end{lemma}

Let us consider the case when we are using an additive, data-independent mechanism for non-adaptive continual observation. That is, $\mathfrak M = \set{\mathcal M: \mathcal M(x) = \counting x + z}$, where $z$ is a random variable over $\mathbb R^T$ whose distribution does not depend on $x$. The proof follows similarly as in the mean-squared case in Edmonds et al.~\cite{edmonds2020power}. 
Note that,    
\begin{align}
\max_{x \in \set{0,1}^T} \mathbb E \sparen{\norm{\mathcal M(x) - \counting x}_\infty^2} = \mathbb E[\norm{z}_\infty^2], \label{eq:exp_mech_inf}
\end{align}
where the expectation is over the coin tosses of $\mathcal M$.

Let $\Sigma = \mathbb E[zz^\top]$ be the covariance matrix of $z$ so that $\mathbb E [\norm{z}_\infty^2] \ge \max_{1 \leq i \leq T} \Sigma[i,i]$.
Now let us define  $K=\counting B_1^T$ to be the so called {\em sensitivity polytope} and $B_1$ to be the $T$-dimensional $\ell_1$ unit ball. 
As $\counting$ has full rank, it follows that $K$ is full dimensional.
Now using Lemma 27 in \citet{edmonds2020power}, we have that there exists an absolute constant $C$ such that
\[
\max_{y \in K} \norm{\Sigma^{-1/2}y}_2 \leq C \epsilon.
\]

Define $L = \Sigma^{1/2}$ and $R = \Sigma^{-1/2}\counting$. Then
\[
\norm{R}_{1 \to 2} = \max_{1 \leq i \leq T} \norm{\Sigma^{-1/2} \counting[:i]}_2 \leq  \max_{y \in K}\norm{\Sigma^{-1/2}y}_2.
\]

That is, $\norm{R}_{1 \to 2} \leq C\epsilon.$ Further, 
\begin{align*}
\norm{L}_{2 \to \infty}^2 &= \max_{1 \leq i \leq T} {(L^\top L)[i,i]} = \max_{1\leq i \leq T} {\Sigma[i,i]} \le  {\mathbb E[\norm{z}_\infty^2]}. 
\end{align*}

By the definition of $\norm{\counting}_\cb$, we thus have
\[
\norm{\counting}_\cb^2 \leq \norm{L}_{2 \to \infty}^2 \norm{R}_{1 \to 2}^2 \leq C^2\epsilon^2 \mathbb E[\norm{z}_\infty^2].
\]
Using the lower bound on $\norm{\counting}_\cb$, rearranging the last inequality, plugging them into \cref{eq:exp_mech_inf}, and using \cite[Lemma 29]{edmonds2020power} completes the proof of \Cref{thm:lower_priv_counting}.
\end{proof}


\section{Experiments}
\label{sec:experiments}
We empirically evaluated algorithms for two problems, namely continual counting and continual  top-1 statistic, i.e.~the frequency of the most-frequent element in the histogram. 
As the focus of this paper is on specifying the exact constants beyond the asymptotic error bound on differentially private counting,  we do not make any assumption on the data or perform no post-processing on the output. 
The main goal of our experiments is to compare the additive error of (1) our mechanism and (2) the binary mechanism instantiated with Gaussian noise (i.e., the binary mechanism that achieves $(\epsilon,\delta)$-differential privacy).
We do not compare our mechanism with the binary mechanism with Laplacian noise as it achieves a stronger notion of differential privacy, namely $\epsilon$-differential privacy, and has an asymptotically worse additive error.

We also implemented Honaker's variant of the binary mechanism~\cite{honaker2015efficient}, but it was so slow that the largest $T$ value for which it terminated before reaching a 5 hour time limit was 512. For these small values of $T$ its $\ell_{\infty}$-error was worse than the $\ell_{\infty}$-error of the binary mechanism with Gaussian noise, which is not surprising as Honaker's variant is optimized to minimize the $\ell_{2}$-error, not the $\ell_{\infty}$-error. Thus, we omit this algorithm in our discussion below. 

\noindent \textbf{Data sets for continual counting.}
For 8 different values of $p$, namely for every
$$p \in \set{2^{-4}, 2^{-5}, 2^{-6}, 2^{-7}, 2^{-8}, 2^{-9}, 2^{-10}, 0 },$$ we generated a stream of $T = 2^{16}$
 Bernoulli random variables $\mathsf{Ber}(p)$. Note that the eighth stream is an all-zero stream.
 Using Bernoulli random variables with $p \ne 0$ ensures that our data streams do not satisfy any smoothness properties, i.e.,~it makes it challenging for the mechanism to output smooth results.

We conjectured that using different $p$ values would not affect  the \emph{magnitude of the additive $\ell_{\infty}$-error}, as the noise of both algorithms is independent of the data. Indeed our experiments confirmed this conjecture, i.e., the additive error in the output is not influenced by the different input streams. Note that the same argument also applies to real-world data, i.e., we would get the same results with real-world data.
%
This was observed before and exploited in the empirical evaluation of differentially private algorithms for industrial deployment. For example, Apple used
an all-zero stream to test its first differentially private algorithm, see the discussion on the accuracy analysis  in the talk of Thakurta at Usenix 2017~\cite{thakurta2017differential}. An all-zero stream was also used in  the empirical evaluation of the differentially private continual binary counting mechanism in ~\citet{mcmahan2022private} (see Figure 1 in the cited paper). 

We evaluated  data streams with varying values of $p$ to  not only study the additive error, but also the \emph{signal to noise ratio (SNR)} in data streams with different sparsity of ones. 

\noindent \textbf{Data sets for top-1 statistic.}
We generated a stream of $2048$ elements from a universe of 20 items using Zipf's law~\cite{zipf2016human}. Zipf’s Law is a statistical distribution that models many data sources that occur in nature, for example, linguistic corpi, in which the frequencies of certain words are inversely proportional to their ranks. This is one of the standard distributions used in estimating the error of differentially private histogram estimation~\cite{cardoso2021differentially}. 

\noindent \textbf{Experimental setup.}
To ensure that the confidence in our estimates is as high as possible  and reduce the fluctuation due to the stochasticity of the Gaussian samples,
we ran both the binary tree mechanism and our matrix mechanism for $10^6$ repetitions and took the average of the outputs  of these executions. 

 \begin{figure}[t]
\centering     
\caption{Comparison of our mechanism with binary  mechanism for $T=2^{16}, \epsilon=0.5,\delta=10^{-10}$ and various sparsity level. The  $x$-axis is the current time epoch, the $y$-axis gives the output of the algorithms at each time epoch.}
{
}
{
\includegraphics[scale=0.28]
{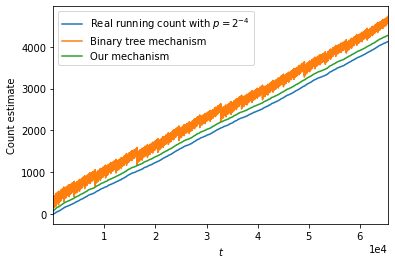}
}
{
\includegraphics[scale=0.28]
{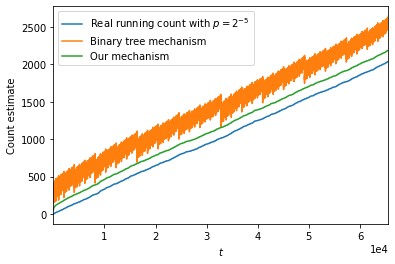}
}
{
\includegraphics[scale=0.28]
{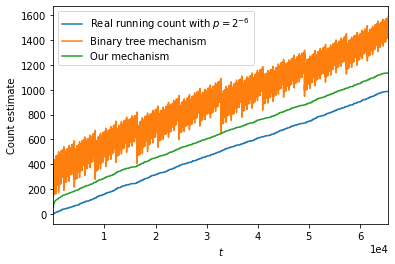}
}
{
\includegraphics[scale=0.28]
{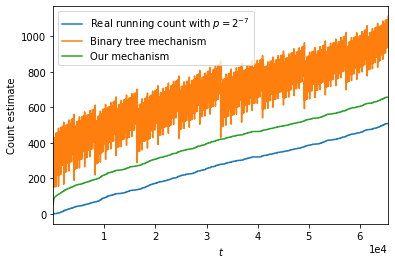}
}
{
\includegraphics[scale=0.28]
{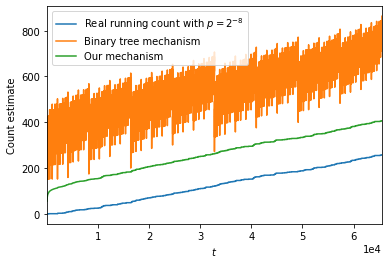}
}
{
\includegraphics[scale=0.28]
{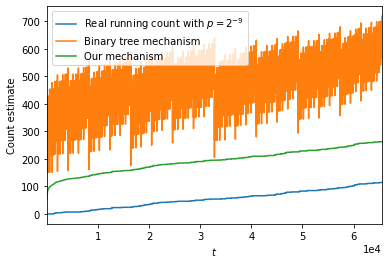}
}
{
\includegraphics[scale=0.28]
{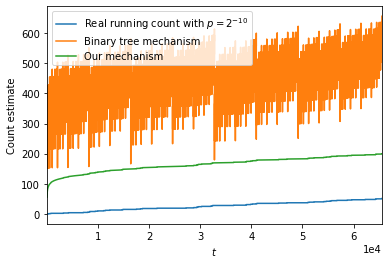}
}
{
\includegraphics[scale=0.28]
{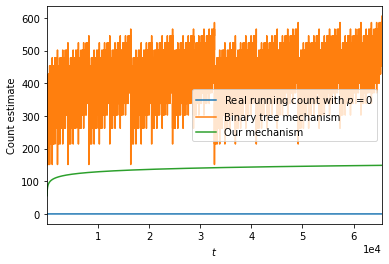}
}
\label{fig:running_count}
\end{figure}

\noindent \textbf{Results.}  \Cref{fig:running_count} shows the output of the algorithms for continual counting.
\Cref{fig:histogram_estimate} shows the private estimate of the frequency of the most frequent item for each algorithm. On the $y$-axis, we report the output of the algorithms and the non-private output, i.e., the real count. The $x$-axis is the current time epoch. 

(1) The first main takeaway  is that our additive error (i.e., difference between our estimate and real count) is consistently less than that of the binary mechanism. For $t=2^i-1$ for $i \in \mathbb N$ the improvement in the additive error is factor of roughly 4. This aligns with our \Cref{remark:suboptimality}. 
We note that a similar observation was made in the recent work~\cite{henzinger2022almost} with respect to the $\ell_2$-error. 

(2) The second main takeway of our experiments is that the error due to the binary mechanism is a non-smooth and a non-monotonic function of $t$ while our mechanism distributes the error smoothly and it is monotonically increasing. This can be explained from the fact that the variance of the noise in our algorithm is a monotonic function, $\ln^2(t)$, while that of the binary mechanism is $\ln(b)$, where $b$ is the number of ones in the bitwise representation of $t$, i.e., a non-smooth, non-monotonic function.

  \begin{table}[h]
\centering
      \begin{tabular}{cccccccc}
      \toprule
      $p$ & $2^{-4}$ & $2^{-5}$ & $ 2^{-6}$ & $2^{-7}$ & $2^{-8}$ & $2^{-9}$ & $2^{-10}$ \\ \midrule
      Binary & 4.72 & 2.40 & 1.12 & 0.61 & 0.31 & 0.15 &
 0.072 \\
     Our mech & 14.50 &  7.37 &  3.46  &  1.86 &  0.96 & 0.47 &
  0.22 \\ \bottomrule   
      \end{tabular}
      \caption{Average signal to noise ratio between private estimates and true count for various sparsity level.}
      \label{tab:average_gap}
  \end{table}
(3) In \Cref{tab:average_gap}, we present the average SNR  over all time epochs between the private estimates and the true count for different sparsity values of the stream. We notice that our output is consistently better and is about three times better when $p=2^{-10}$. 
We noticed that for $\epsilon=0.5,\delta=10^{-10}$ when the fraction of ones is about $1/80$, the average SNR  for the binary mechanism drops below $1$, i.e., the error is larger than the true count, while  for our mechanism it only drops below $1$ when the fraction of ones is $\leq 1/250$. That is, our mechanism can handle three times sparser streams than the binary mechanism for the same SNR. This observation continued to hold for different privacy parameters.

\begin{figure}[t]
\centering
\caption{\small{ (Left) We give the Zipf’s law distribution that items are sampled from at each round of an event stream. (Right) the running estimate of the most frequent item using the binary mechanism and our mechanism instantiated with $\epsilon=0.1, \delta = 10^{-10}$.}}
\includegraphics[scale=0.3]{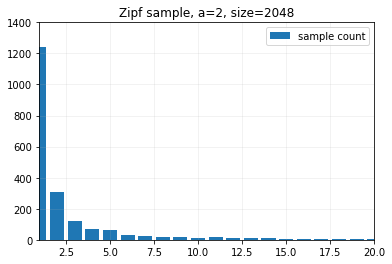}
\includegraphics[scale=0.3]{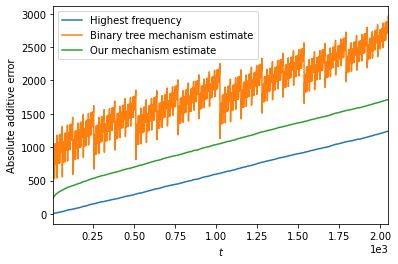}
    \label{fig:histogram_estimate}
\vspace{-8mm}
\end{figure}

(4) For histogram estimation, our experiments reveal that our mechanism performs consistently better than the binary mechanism, both in terms of the absolute value of the additive error incurred as well as in terms of the smoothness of the error. This is consistent with our theoretical results. Further, on average over all time epochs, the SNR for our mechanism is $1.47$ while that of binary mechanism is $0.52$, i.e., it is a factor of about 3 better. 

\section{Conclusion}
\label{sec:conclusion}

In this paper, we study the problem of binary counting under continual release. The motivation for this work is that (1)  only an asymptotic analysis is known for the additive error for the classic mechanism (known as the binary mechanism) for binary counting under continual release, and (2)  in practice the additive error is very non-smooth, which hampers its practical usefulness. 
Thus, we ask 

\vspace{-3mm}
\begin{quote}
{\em Is it possible to design  differentially private algorithms with fine-grained  bounds on the constants of the additive  error?}
\end{quote}
\vspace{-3mm}

We first observe that  the matrix mechanism can actually be used for binary counting in the continual release model \emph{if the factorization uses lower-triangular matrices.} Then we give an explicit factorization for  $\counting$ that fulfill the following properties:

(1) We improved a 28 years old result on $\norm{\counting}_\cb$ to give an analysis of the additive error that only has a small gap between the upper and lower bound for the counting problem. This means that the behavior of the additive error is very well understood.

(2) The additive error is a monotonic smooth function of the number of updates performed so far. In contrast, previous algorithms would either output with the error that changes non-smoothly over time, making them less interpretable and reliable. 

(3) The factorization for the binary mechanism consists of two lower-triangular matrices with exactly $T$ distinct non-zero entries that follow a simple pattern so that only $O(T)$ space is needed.

(4) We show that all these properties are not just theoretical advantages, but also makes a big difference in practice (see \Cref{fig:running_count}).

(5) Our algorithm is very simple to implement, consisting of a matrix-vector multiplication and the addition of two vectors. Simplicity is an important design principle in large-scale deployment due to one of the important goals, which is to reduce the points of vulnerabilities in  a system. 
As there is no known technique to verify whether a system is indeed $(\epsilon,\delta)$-differentially private, it is important to ensure that a deployed system faithfully implements a given algorithm that has provable guarantee. This is one main reason for us to pick the Gaussian mechanism: it is easy to implement with floating point arithmetic while maintaining the provable guarantee of privacy.  Further, the privacy guarantee can be easily stated in the terms of concentrated-DP or Renyi-DP. 

Finally, we show that our bounds have diverse applications that range from binary counting to maintaining histograms,  various graph functions, outputting a synthetic graph that maintains the value of all cuts, substring counting, and episode counting. 
We believe that there are more applications of our mechanism, and this work will bring more attention to leading constants in the analysis of differentially private algorithms.

\subsection*{Acknowledgements}
 This project has received funding from the European Research Council (ERC)
  under the European Union's Horizon 2020 research and innovation programme (Grant agreement 
   \begin{wrapfigure}{r}{0.15\textwidth}\label{fig:diff}
  \includegraphics[width=0.13\textwidth]{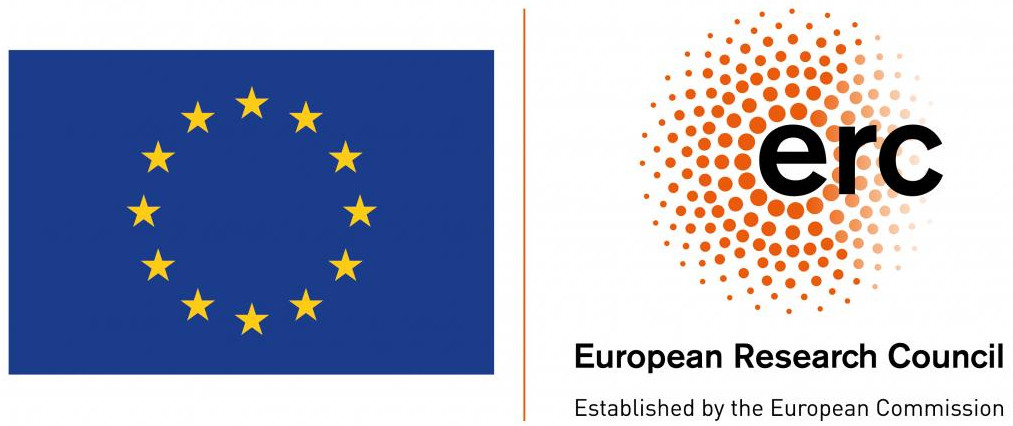}
\end{wrapfigure}
 No. 101019564 ``The Design of Modern Fully Dynamic Data Structures (MoDynStruct)'' and from the Austrian Science Fund (FWF) project ``Fast Algorithms for a Reactive Network Layer (ReactNet)'', P~33775-N, with additional funding from the \textit{netidee SCIENCE Stiftung}, 2020--2024. 
JU's research was funded by Decanal Research Grant.  The authors would like to thank Rajat Bhatia, Aleksandar Nikolov, Rasmus Pagh, Vern Paulsen, Ryan Rogers, Thomas Steinke, Abhradeep Thakurta, and Sarvagya Upadhyay  for useful discussions.



\newcommand{\etalchar}[1]{$^{#1}$}

\newpage
\appendix
\onecolumn
\end{document}